\documentclass[11pt,letterpaper,showkeys]{article}
\usepackage{amsmath,amsthm}
\usepackage{amsfonts}

\newtheorem{theorem}{Theorem}[section]
\newtheorem{lemma}[theorem]{Lemma}
\newtheorem{corollary}[theorem]{Corollary}
\newtheorem{proposition}[theorem]{Proposition}
\newtheorem{remark}[theorem]{Remark}

\newtheorem{definition}{Definition}

\newcommand{\ba}{\begin{array}}
\newcommand{\ea}{\end{array}}
\newcommand{\beq}{\begin{equation}}
\newcommand{\eeq}{\end{equation}}
\newcommand{\beqa}{\begin{eqnarray}}
\newcommand{\eeqa}{\end{eqnarray}}
\newcommand{\beqas}{\begin{eqnarray*}}
\newcommand{\eeqas}{\end{eqnarray*}}
\newcommand{\bi}{\begin{itemize}}
\newcommand{\ei}{\end{itemize}}
\newcommand{\gap}{\hspace*{2em}}

\newcommand{\nn}{\nonumber}

\newcommand{\simless}[2]{\Delta^{#2}_{\le}(#1)}
\newcommand{\simequal}[2]{\Delta^{#2}_{=}(#1)}
\newcommand{\Si}[1]{{\cal S}^{#1}}
\newcommand{\Sip}[1]{{\cal S}^{#1}_{+}}

\def\vgap{\vspace*{.1in}}

\def\eqnok#1{(\ref{#1})}
\def\argmin{{\rm argmin}}

\def\QED{\ifhmode\unskip\nobreak\fi\ifmmode\ifinner\else\hskip5pt\fi\fi
  \hbox{\hskip5pt\vrule width5pt height5pt depth1.5pt\hskip1pt}}

\def\aa{{\bf a}}
\def\bu{\bullet}
\def\bb{{\bf b}}
\def\ee{{\bf e}}

\def\bX{{\bar X}}
\def\bx{{x^0}}

\def\cC{{\cal C}}
\def\cE{{\cal E}}
\def\cG{{\cal G}}
\def\cL{{\cal L}}
\def\cN{{\cal N}}

\def\cS{{\cal S}}
\def\cB{{\cal B}}
\def\cU{{\cal U}}
\def\cV{{\cal V}}

\def\eps{{\epsilon}}
\def\fg{{f_{\gamma}}}
\def\diag{{\rm Diag}}
\def\F{{\rm F}}

\def\Sn{{\cal S}^n}

\def\Snpp{{\cal S}^n_{++}}

\def\trx{{\tilde r_x}}

\def\tOmega{{\tilde \Omega}}
\def\tr{{\rm Tr}}
\def\tu{{\tilde u}}
\def\tU{{\tilde U}}
\def\tV{{\tilde V}}

\def\vec{{\rm vec}}
\def\K{{H}}
\def\setU{{U}}

\def\setX{{X}}

\def\vsigma{{\sigma_V}}
\def\xeps{{x_{\epsilon}^\gamma}}
\def\xleps{{x_{\epsilon}^\lambda}}

\def\uh{{p_U}}
\def\dU{{p_U}}
\def\usigma{{\sigma_U}}
\def\DU{{D_U}}

\title{Convex Optimization Methods for Dimension Reduction and Coefficient 
Estimation in Multivariate Linear Regression}
\author{
Zhaosong Lu%
\thanks{Department of Mathematics, Simon Fraser University, Burnaby, 
BC V5A 1S6, Canada (Email: {\tt zhaosong@sfu.ca}).
This author was supported in part by SFU President's Research Grant
and NSERC Discovery Grant.}
\and 
Renato D. C. Monteiro%
\thanks{School of Industrial and Systems 
Engineering, Georgia Institute of Technology, Atlanta, GA
30332-0205, USA (Email: {\tt monteiro@isye.gatech.edu}).
This author was supported in part by 
NSF Grants CCF-0430644 and CCF-0808863 and
ONR Grants N00014-05-1-0183 and N00014-08-1-0033.}
\and
Ming Yuan%
\thanks{School of Industrial and Systems Engineering, Georgia 
Institute of Technology, Atlanta, GA 30332-0205, USA 
(Email: {\tt myuan@isye.gatech.edu}). This author was supported in 
part by NSF Grants DMS-0624841 and DMS-0706724.}
}

\date{January 10, 2008 (Revised: March 6, 2009)}

\begin{document}   

\maketitle

\begin{abstract}
In this paper, we study convex optimization methods for computing the 
trace norm regularized least squares estimate in multivariate linear 
regression. The so-called factor estimation and selection (FES) method, 
recently proposed by Yuan et al.\ \cite{YuEkLuMo07}, conducts parameter 
estimation and factor selection simultaneously and have been shown to 
enjoy nice properties in both large and finite samples. To compute the 
estimates, however, can be very challenging in practice because of the 
high dimensionality and the trace norm constraint. In this paper, we 
explore a variant of Nesterov's smooth method \cite{tseng08} and 
interior point methods for computing the penalized least squares estimate. 
The performance of these methods is then compared using a set of randomly 
generated instances. We show that the variant of Nesterov's smooth method 
\cite{tseng08} generally outperforms the interior point method implemented 
in SDPT3 version 4.0 (beta) \cite{ToTuTo06-1} substantially . Moreover, 
the former method is much more memory efficient. 

\vskip14pt


\vskip14pt


\end{abstract}

\section{Introduction}

Multivariate linear regression is routinely used in statistics to model
the predictive relationships of multiple related responses on a common
set of predictors. In general multivariate linear regression, we have
$l$ observations on $q$ responses $\bb=(b_1, \ldots, b_q)'$ and $p$
explanatory variables $\aa=(a_1,\ldots,a_p)'$, and
\beq \label{multi-reg1}
B=A U + E,
\eeq
where $B=(\bb^1,\ldots,\bb^l)'\in \Re^{l\times q}$ and $A=(\aa^1,
\ldots,\aa^l)'\in\Re^{l\times p}$ consists of the data of responses
and explanatory variables, respectively, $U\in \Re^{p \times q}$
is the coefficient matrix, $E=(\ee^1,\ldots,\ee^l)'\in\Re^{l\times q}$
is the regression noise, and all $\ee^i$s are independently sampled
from $\cN(0, \Sigma)$.

Classical estimators for the coefficient matrix $U$ such as the least
squares estimate are known to perform sub-optimally because they do not
utilize the information that the responses are related. This problem
is exacerbated when the dimensionality $p$ or $q$ is moderate or large.
Linear factor models are widely used to overcome this problem. In the
linear factor model, the response $B$ is regressed against a small
number of linearly transformed explanatory variables, which are often referred
to as factors. More specifically, the linear factor model can be expressed
as
\beq \label{multi-reg2}
B=F\Omega+E,
\eeq
where $\Omega\in \Re^{r \times q}$, and $F=A\Gamma$ for some
$\Gamma\in \Re^{p \times r}$ and $r\le \min\{p,q\}$. The columns of $F$,
namely, $F_j \ (j=1,\ldots, r)$ represent the so-called factors. Clearly
\eqnok{multi-reg2} is an alternative representation of \eqnok{multi-reg1}
with $U=\Gamma\Omega$, and the dimension of the estimation problem reduces
as $r$ decreases.
Many popular methods including canonical correction (Hotelling \cite{Hot35, Hot36}),
reduced rank (Anderson \cite{Ander51}, Izenman \cite{Ize75}, Reinsel and Velu
\cite{ReVe98}), principal components (Massy \cite{Mas65}), partial least squares
(Wold \cite{Wol75}) and joint continuum regression (Brooks and Stone \cite{BrSt94})
among others can all be formulated in the form of linear factor
regression. They differ in the way in which the factors are determined.

Given the number of factors $r$, estimation in the linear factor model most often 
proceeds in two steps: the factors, or equivalently $\Gamma$, are first constructed,
and then $\Omega$ is estimated by least squares for \eqnok{multi-reg2}.
It is obviously of great importance to be able to determine $r$ for \eqnok{multi-reg2}. 
For a smaller number of factors, a more accurate estimate is expected since there 
are fewer free parameters. But too few factors may not be sufficient to describe 
the predictive relationships. In all of the aforementioned methods, the number of 
factors $r$ is chosen in a separate step from the estimation of \eqnok{multi-reg2}
through either hypothesis testing or cross-validation. The coefficient
matrix is typically estimated on the basis of the number of factors
selected. Due to its discrete nature, this type of procedure can be very
unstable in the sense of Breiman \cite{Bre96}: small changes in the data
can result in very different estimates.

Recently, Yuan et al.\ \cite{YuEkLuMo07} proposed a novel method that can
simultaneously choose the number of factors, determine the factors and estimate
the factor loading matrix $\Omega$. It has been demonstrated that the so-called
factor estimation and selection (FES) method combines and retains the advantages
of the existing methods. FES is a constrained least square estimate where the trace 
norm or the nuclear norm (or the Ky Fan $m$-norm where $m :=\min\{p,q\}$) of the 
coefficient matrix $U$ is forced to be smaller than an upper bound:
\beq \label{multi-reg3}
\ba{ll}
\min\limits_{U} & \tr((B-AU)W(B-AU)') \\
\mbox{s.t.}   & \sum\limits_{i=1}^m \sigma_i(U) \le M.
\ea
\eeq
where $W$ is a positive definite weight matrix. Common choices of the
weight matrix $W$ include $\Sigma^{-1}$ and $I$. To fix ideas, we assume
throughout the paper that $W=I$. Under this assumption,
\eqnok{multi-reg3} is equivalent to
\beq \label{multi-reg4}
\ba{ll}
\min\limits_{U} & \|B-AU\|^2_F \\
\mbox{s.t.}   & \sum\limits_{i=1}^m \sigma_i(U) \le M.
\ea
\eeq
It is shown in Yuan et al.\ \cite{YuEkLuMo07} that the constraint used by FES
encourages sparsity in the factor space and at the same time gives shrinkage
coefficient estimates and thus conducts dimension reduction and estimation
simultaneously in the multivariate linear model. 
Recently, Bach \cite{Bach08} further provided necessary and sufficient conditions for 
rank consistency of trace norm minimization with the square loss by considering 
the Lagrangian relaxation of \eqnok{multi-reg4}. He also proposed a Newton-type method 
for finding an approximate solution to the latter problem. It shall be mentioned 
that his method is only suitable for the problems where $p$ and $q$ are not too 
large.  

In addition, the trace norm relaxation has been used in literature for rank 
minimization problem. In particular, Fazel et al.\ \cite{FaHiBo01} considered minimizing 
the rank of a matrix $U$ subject to $U \in \cC$, where $\cC$ is a closed convex set. 
They proposed a convex relaxation to this problem by replacing the rank of $U$ by the 
trace norm of $U$. Recently, Recht et al.\ \cite{ReFaPa07} showed that under some 
suitable conditions, such a convex relaxation is tight when $\cC$ is an affine 
manifold. The authors of \cite{ReFaPa07} also discussed some first- and second-order 
optimization methods for solving the trace norm relaxation problem.    

The goal of this paper is to explore convex optimization methods, namely, a variant of 
Nesterov's smooth method \cite{tseng08}, and
interior point methods for solving \eqnok{multi-reg4}.
We also compare the performance of these methods on a set of randomly
generated instances. We show that the variant of Nesterov's smooth method \cite{tseng08} 
generally outperforms the interior point method implemented in the code SDPT3 version 
4.0 (beta) \cite{ToTuTo06-1} substantially, and that the former method 
requires much less memory than the latter one.

The rest of this paper is organized as follows. In Subsection \ref{notation},
we introduce the notation that is used throughout the paper.
In Section \ref{res-eig_singv}, we present some technical results that
are used in our presentation. In Section \ref{reform}, we provide a
simplification for problem \eqnok{multi-reg4}, and present
cone programming and smooth saddle point reformulations for it.
In Section \ref{numerical}, we review a variant of Nesterov's smooth method 
\cite{tseng08} and discuss the details of its implementation for solving the 
aforementioned smooth saddle point reformulations of \eqnok{multi-reg4}.
In Section \ref{comp}, we present computational results comparing
a well-known second-order interior-point
method applied to the aforementioned cone programming reformulations
of \eqnok{multi-reg4} with the variant of Nesterov's smooth method for solving smooth 
saddle point reformulations of \eqnok{multi-reg4}. Finally, 
we present some concluding remarks in Section \ref{concl-remark}
and state some additional technical results in the Appendix.

\subsection{Notation}
\label{notation}

The following notation is used throughout our paper. For any real 
number $\alpha$, $[\alpha]^+$ denotes the nonnegative part of $\alpha$,
that is, $[\alpha]^+ = \max\{\alpha,0\}$.
The symbol $\Re^p$ denotes the $p$-dimensional 
Euclidean space. We denote by $e$ the vector of all ones whose 
dimension should be clear from the context.
For any $w\in\Re^p$, $\diag(w)$ denotes the 
$p \times p$ diagonal matrix whose $i$th diagonal element
is $w_i$ for $i=1,\ldots,p$. The Euclidean norm in $\Re^p$ is
denoted by $\|\cdot\|$.

We let $\Si{n}$ denote the space of
$n\times n$ symmetric matrices, and $Z  \succeq 0$
indicate that $Z$ is positive semidefinite. We also write
$\Sip{n}$ for $\{Z \in \Sn : Z \succeq 0 \}$, and $\Snpp$ for its
interior, the set of positive definite matrices in $\Sn$.
For any $Z \in \Si{n}$, we let $\lambda_i(Z)$, for $i=1,...,n$, denote
the $i$th largest eigenvalue of $Z$,
$\lambda_{\min}(Z)$ (resp., $\lambda_{\max}(Z)$)
denote the minimal (resp., maximal) eigenvalue of $Z$, and define
$\|Z\|_{\infty} := \max_{1\le i\le n} |\lambda_i(Z)|$ and 
$\|Z\|_1=\sum_{i=1}^n |\lambda_i(Z)|$.
Either the identity matrix or operator will be denoted by $I$.

The space of all $p \times q$ matrices with real entries 
is denoted by $\Re^{p \times q}$.  Given matrices $X$ and $Y$ in 
$\Re^{p \times q}$, the standard inner product is defined by 
$X \bullet Y = \tr (X^TY)$, where $\tr(\cdot)$ denotes the trace of a 
matrix.
The operator norm and the Frobenius norm
of a $p \times q$-matrix $X$ are defined as
$\|X\| := \max \{\|Xu\| : \|u\| \le 1 \} = [\lambda_{\max}(X^TX)]^{1/2}$
and $\|X\|_F := \sqrt{X \bullet X}$, respectively.
Given any $X\in\Re^{p \times q}$, 
we let $\vec(X)$ denote the vector in $\Re^{pq}$ obtained
by stacking the columns of $X$ according to the order in which
they appear in $X$, and $\sigma_i(X)$ denote the $i$th largest
singular value of $X$ for $i=1,\ldots, \min\{p,q\}$.
(Recall that $\sigma_i(X)=[\lambda_i(X^TX)]^{1/2}
= [\lambda_i(XX^T)]^{1/2}$ for $i=1,\ldots, \min\{p,q\}$.)
Also, let $\cG: \Re^{p \times q} \to \Re^{(p+q) \times (p+q)}$ be  
defined as 
\beq \label{Gx}
\cG(X) := \left(\ba{ll}
0 & X^T \\
X & 0
\ea \right),  \ \forall X \in \Re^{p \times q}.
\eeq

The following sets are used throughout the paper:
\beqas
\cB_F^{p \times q}(r) &:=& \{ X \in \Re^{p \times q} : \|X\|_F \le r \}, \\
\simequal{r}{n} &:=& \{ Z \in \Si{n} : \|Z\|_1 = r, \, Z \succeq 0 \}, \\
\simless{r}{n} &:=& \{ Z \in \Si{n} : \|Z\|_1 \le r, \, Z \succeq 0 \}, \\
\cL^p &:=& \left\{x\in \Re^p : x_1 \geq \sqrt{x^2_2+\ldots+x^2_p} \, \right\},
\eeqas
where the latter is the well-known $p$-dimensional second-order cone.

Let $\cU$ be a normed vector space whose norm is denoted by $\|\cdot\|_\cU$.
The dual space of $\cU$, denoted by $\cU^*$, is the normed vector
space consisting of all linear functionals of $u^*: \cU \to \Re$, endowed 
with the dual norm $\|\cdot\|_\cU^*$ defined as
\[
\|u^*\|_\cU^* = \max\limits_u \{\langle u^*, u \rangle: \ \|u\|_\cU \le 1 \},
\ \ \ \forall u^*\in \cU^*, 
\]
where $\langle u^*,u\rangle := u^*(u)$ is the value of
the linear functional $u^*$ at $u$.

If $\cV$ denotes another normed vector space with norm $\|\cdot\|_\cV$, and
$\cE:\cU \to \cV^*$ is a linear operator, 
the operator norm of $\cE$ is defined as 
\beq \label{normC1}
\|\cE\|_{\cU,\cV} = \max\limits_u\{\|\cE u\|_\cV^*: \|u\|_\cU \le 1\}.
\eeq
A function $f: \Omega \subseteq U \to \Re$ is said to be
$L$-Lipschitz-differentiable
with respect to $\|\cdot\|_\cU$ if it is differentiable and
\beq \label{sete}
\|\nabla f(u) - \nabla f(\tilde u)\|_\cU^* \le L \|u - \tilde u\|_\cU,
\ \ \ \forall u, \tilde u \in \Omega. 
\eeq

\section{Some results on eigenvalues and singular values}
\label{res-eig_singv}

In this subsection, we establish some technical results about
eigenvalues and singular values which will be used in our presentation.

The first result gives some well-known identities involving the
maximum eigenvalue of a real symmetric matrix.

\begin{lemma} \label{eigprop}
For any $Z\in \cS^n$ and scalars $\alpha>0$ and $\beta \in \Re$,
the following statements hold:
\beqa
\lambda_{\max}(Z) &=& \max\limits_{W \in \simequal{1}{n}} Z \bullet W,
\label{eigprop3} \\
\left[ \alpha \lambda_{\max}(Z) + \beta \right]^+ &=&
\max\limits_{W \in \simless{1}{n}} \alpha Z \bullet W + \beta \tr(W).
\label{eigprop1}
\eeqa
\end{lemma} 

\begin{proof}
Identity \eqnok{eigprop3} is well-known. 
We have
\beqas
\left[\alpha \lambda_{\max}(Z) + \beta \right]^+ &=&
\left[\lambda_{\max}(\alpha Z + \beta I) \right]^+ \ = \
\max_{t \in [0,1]} t \lambda_{\max}(\alpha Z + \beta I) \\
&=& \max_{t \in [0,1], W \in \simequal{1}{n}}
t (\alpha Z + \beta I) \bullet W
\ = \ \max_{W \in \simless{1}{n}} (\alpha Z + \beta I) \bullet W,
\eeqas
where the third equality is due to \eqnok{eigprop3} and the fourth
equality is due to the fact that
$tW$ takes all possible values in $\simless{1}{n}$ under the
condition that $t \in [0,1]$ and $W \in \simequal{1}{n}$.
\end{proof}

\vgap

The second result gives some characterizations of the sum of the
$k$ largest eigenvalues of a real symmetric matrix.

\begin{lemma} \label{sumeig-prop}
Let $Z \in \cS^n$ and integer $1 \le k \le n$ be given.
Then, the following statements hold:
\begin{itemize}
\item[a)]
For $t\in \Re$, we have
\[
\sum\limits_{i=1}^k \lambda_i(Z) \le t  \ \Leftrightarrow \
\left\{ \ba{lcl}
t- ks - \tr(Y) & \ge & 0,  \\
Y -Z + sI & \succeq & 0, \\
Y & \succeq & 0,
\ea \right.
\]
for some $Y \in \Si{n}$ and $s\in\Re$;
\item[b)] The following identities hold:
\beqa
\sum\limits_{i=1}^k \lambda_i(Z)
&=& \min\limits_{Y \in \Sip{n}} \max_{W \in \simequal{1}{n}}
k (Z - Y) \bullet W + \tr(Y) \label{cchar1} \\
&=& \max\limits_{W \in \Si{n}}
\{ Z \bullet W : \tr(W)=k, \, 0 \preceq W \preceq I\}. \label{cchar2}
\eeqa
\item[c)]
For every scalar $\alpha>0$ and $\beta \in \Re$,
the following identities hold:
\beqa
\left[ \alpha \sum\limits_{i=1}^k \lambda_i(Z) + \beta \right]^+
&=& \min\limits_{Y \in \Sip{n}} \max\limits_{W \in \simless{1}{n}}
k (\alpha Z - Y) \bullet W + [ \beta + \tr(Y)] \, \tr(W)
\label{cchar3} \\
&=& \max\limits_{W \in \Si{n},\,t \in \Re} \,
\left\{ \alpha Z \bullet W + \beta t :
\tr(W)=t k, \, 0 \preceq W \preceq t I, \, 0 \le t \le 1
\right\}. \label{cchar4}
\eeqa
\end{itemize}
\end{lemma}

\begin{proof}
a) This statement is proved on pages 147-148 of Ben-Tal and Nemirovski
\cite{BenNem01}.

b) Statement (a) clearly implies that
\beq \label{ttt}
\sum\limits_{i=1}^k \lambda_i(Z) =
\min\limits_{s \in \Re, Y \in \Si{n}}
\{ k s + \tr(Y): \ Y + sI \succeq Z, \, Y \succeq 0 \}.
\eeq
Noting that the condition $Y + sI \succeq Z$ is equivalent to
$s \ge \lambda_{\max}(Z - Y)$, we can eliminate the variable
$s$ from the above min problem to conclude that
\beq \label{auxx}
\sum\limits_{i=1}^k \lambda_i(Z) =
\min \, \{ k \lambda_{\max}(Z-Y) + \tr(Y) : Y \in \Sip{n} \}.
\eeq
This relation together with \eqnok{eigprop3} clearly
implies identity \eqnok{cchar1}.
Moreover, noting that the max problem \eqnok{cchar2} is the dual of
min problem \eqnok{ttt} and that they both have strictly feasible solutions,
we conclude that identity \eqnok{cchar2} holds in view of
a well-known strong duality result.

c) Using \eqnok{auxx}, the fact that $\inf_{x \in X} [x]^+ = [\inf \, X]^+$
for any $X \subseteq \Re$ and \eqnok{eigprop1}, we obtain
\beqas
\left[ \alpha \sum\limits_{i=1}^k \lambda_i(Z) + \beta \right]^+
&=& \left[ \, \sum\limits_{i=1}^k \lambda_i \left(\alpha Z +
\frac\beta k \, I \right) \, \right]^+  \\ 
&=& \left[ \min\limits_{Y \in \Sip{n}} \, k \lambda_{\max}
\left(\alpha Z + \frac\beta k \, I - Y \right) + \tr(Y) \right]^+  \\ 
&=& \min\limits_{Y \in \Sip{n}} \, \left[ k \lambda_{\max}
\left(\alpha Z + \frac\beta k \, I - Y \right) + \tr(Y) \right]^+ \\
&=& \min\limits_{Y \in \Sip{n}} \, \max\limits_{W \in \simless{1}{n}}
k \left( \alpha Z + \frac\beta k \, I - Y \right) \bullet W + \tr(Y) \tr(W),
\eeqas
from which \eqnok{cchar3} immediately follows. Moreover, using
\eqnok{cchar2}, the fact that $[\gamma]^+=\max_{t \in [0,1]} t \gamma$
for every $\gamma \in \Re$ and performing the change of variable
$Y=t \tilde Y$ in the last equality below, we obtain
\beqas
\left[ \alpha \sum\limits_{i=1}^k \lambda_i(Z) + \beta \right]^+
&=& \left[ \, \sum\limits_{i=1}^k \lambda_i \left(\alpha Z +
\frac\beta k \, I \right) \, \right]^+ \\
&=& \left[ \, \max\limits_{\tilde Y \in \Si{n}} \, \left\{
\left(\alpha Z + \frac\beta k \, I \right) \bullet \tilde Y
: \tr(\tilde Y)=k, \, 0 \preceq \tilde Y \preceq I \right\} \, \right]^+ \\
&=& \max\limits_{\tilde Y \in \Si{n},\,t \in \Re} \,
\left\{ t \left(\alpha Z + \frac\beta k \, I \right) \bullet \tilde Y
: \tr(\tilde Y)=k, \, 0 \preceq \tilde Y \preceq I, \, 0 \le t \le 1 \right\} \\
&=& \max\limits_{Y \in \Si{n},\,t \in \Re} \, \left\{ \left(\alpha Z + \frac\beta k \, I \right)
\bullet Y : \tr(Y)=t k, \, 0 \preceq Y \preceq t I, \, 0 \le t \le 1
\right\},
\eeqas
i.e., \eqnok{cchar4} holds.
\end{proof}

\vgap

\vgap

\begin{lemma} \label{sum-sings}
Let $X \in \Re^{p \times q}$ be given.
Then, the following statements hold: 
\begin{itemize}
\item[a)]
the $p+q$ eigenvalues of the symmetric matrix $\cG(X)$ defined in \eqnok{Gx},
arranged in nonascending order, are
\[
\sigma_1(X), \cdots, \sigma_m(X), 0, \cdots,0,
-\sigma_m(X),\cdots,-\sigma_1(X),
\]
where $m:=\min(p,q)$;
\item[b)] For any positive integer $k \le m$, we have
\[
\sum\limits_{i=1}^k \sigma_i(X) = \sum\limits_{i=1}^k \lambda_i(\cG(X)).
\]
\end{itemize}
\end{lemma}

\begin{proof}
Statement (a) is proved on page 153 of \cite{BenNem01} and statement (b)
is an immediate consequence of (a).
\end{proof}

\vgap

The following result about the sum of the $k$
largest singular values of a matrix follows immediately from
Lemmas \ref{sumeig-prop} and \ref{sum-sings}.

\begin{proposition} \label{singv-sadd}
Let $X \in \Re^{p \times q}$ and integer $1 \le k \le \min\{p,q\}$ be given and set
$n:= p+q$. Then:
\begin{itemize}
\item[a)]
For $t\in \Re$, we have
\[
\sum\limits_{i=1}^k \sigma_i(X) \le t  \ \Leftrightarrow \
\left\{ \ba{lcl}
t- ks - \tr(Y) & \ge & 0,  \\
Y - \cG(X) + sI & \succeq & 0, \\
Y & \succeq & 0,
\ea \right.
\]
for some $Y \in \Si{n}$ and $s\in\Re$;
\item[b)] The following identities hold:
\beqa
\sum\limits_{i=1}^k \sigma_i(X)
&=& \min\limits_{Y \in \Sip{n}} \max_{W \in \simequal{1}{n}}
k (\cG(X) - Y) \bullet W + \tr(Y) \label{char1} \\
&=& \max\limits_{W \in \Si{n}}
\{ \cG(X) \bullet W : \tr(W)=k, \, 0 \preceq W \preceq I\}. \label{char2}
\eeqa
\item[c)]
For every scalar $\alpha>0$ and $\beta \in \Re$,
the following identities hold:
\beqa
\lefteqn{ \left[ \alpha \sum\limits_{i=1}^k \sigma_i(X) + \beta \right]^+
\ = \ \min\limits_{Y \in \Sip{n}} \max\limits_{W \in \simless{1}{n}}
k (\alpha \cG(X) - Y) \bullet W + [ \beta + \tr(Y)] \, \tr(W) }
\label{char3} \\
&=& \max\limits_{W \in \Si{n},\,t \in \Re}
\, \left\{ \alpha \cG(X) \bullet W + \beta t :
\tr(W)=t k, \, 0 \preceq W \preceq t I, \, 0 \le t \le 1
\right\}. \label{char4}
\eeqa
\end{itemize}
\end{proposition}

%

\section{Problem reformulations}
\label{reform}

This section consists of three subsections. The first subsection shows that
the restricted least squares problem \eqnok{multi-reg4} can be reduced to 
one which does not depend on the (usually large) number of rows of the matrices
$A$ and/or $B$. In the second and third subsections, we provide cone programming 
and smooth saddle point reformulations for \eqnok{multi-reg4}, respectively. 


\subsection{Problem simplification}
\label{reduction}

Observe that the number of rows of the data matrices $A$ and $B$ which
appear in \eqnok{multi-reg4} is equal to the number of observations $l$, 
which is usually quite large in many applications.  However, the size 
of the decision variable $U$ in \eqnok{multi-reg4} does not depend on $l$.
In this subsection we show how problem \eqnok{multi-reg4} can be reduced 
to similar types of problems in which the new matrix $A$ is a $p \times p$ 
diagonal matrix and hence to problems which do not depend on $l$. Clearly, 
from a computational point of view, the resulting formulations
need less storage space and can be more efficiently solved.

Since in most applications, the matrix $A$ has full column
rank, we assume that this property holds throughout the paper.
Thus, there exists an orthonormal matrix 
$Q\in \Re^{p\times p}$ and a positive diagonal matrix $\Lambda\in 
\Re^{p \times p}$ such that $A^TA=Q\Lambda^2 Q^T$. Letting
\beq \label{transform}
X := Q^TU, \ \ \ \K := \Lambda^{-1}Q^TA^TB,
\eeq 
we have
\beqas  
\|B-AU\|^2_{\F} - \|B\|_F^2 &=& \|AU\|_F^2 - 2\, (AU) \bullet B
\ = \ \tr(U^TA^TAU) - 2 \, \tr(U^TA^TB) \\
&=& \tr(U^TQ\Lambda^2Q^TU) - 2\, \tr(U^TQ\Lambda\K) \\
&=& \|\Lambda X\|_F^2 - 2 (\Lambda X) \bullet \K
\ = \ \|\Lambda X - \K \|_F^2 - \|\K\|_F^2.
\eeqas
Noting that the singular values of $X=Q^TU$ and $U$ are identical,
we immediately see from the above identity that \eqnok{multi-reg4}
is equivalent to
\beq \label{constr} 
\ba{ll}
\min\limits_{X} & \frac12\|\Lambda X - \K \|^2_{\F} \\
\mbox{s.t.}   & \sum\limits_{i=1}^m \sigma_i(X) \le M, 
\ea 
\eeq
where $\Lambda$ and $\K$ are defined in \eqnok{transform}. 

In view of Theorem \ref{lagr-eps}, we observe that for any 
$\lambda \ge 0$ and $\epsilon \ge 0$, any $\eps$-optimal solution 
$X_\eps$ of the following Lagrangian relaxation problem  
\beq \label{unconstr} 
\min\limits_X \, \frac12\|\Lambda X - \K \|^2_{\F} +
\lambda \sum\limits_{i=1}^m \sigma_i(X).
\eeq
is an $\eps$-optimal solution of problem \eqnok{constr} with 
$M=\sum_{i=1}^m\sigma_i(X_\eps)$. 
In practice, we often need to solve problem \eqnok{constr} for a 
sequence of $M$ values. Hence, one way to solve such problems is to solve 
problem \eqnok{unconstr} for a sequence of $\lambda$ values.

We will later present convex optimization methods 
for approximately solving the formulations \eqnok{constr} and \eqnok{unconstr}, 
and hence, as a by-product, formulation \eqnok{multi-reg4}.

Before ending this subsection, we provide bounds on the
optimal solutions of problems \eqnok{constr} and \eqnok{unconstr}.


\begin{lemma} \label{constr-boundX}
For every $M>0$, problem \eqnok{constr} has a unique optimal solution
$X_M^*$. Moreover,
\beq \label{trX}
\|X_M^*\|_F \leq \tilde r_x := \min\left\{\frac{2\|\Lambda \K\|_F}
{\lambda^2_{\min}(\Lambda)},M\right\}.
\eeq
\end{lemma}

\begin{proof}
Using the fact that $\Lambda $ is a $p \times p$ positive diagonal matrix, 
it is easy to see that the objective function of \eqnok{constr} is a 
(quadratic) strongly convex function, from which we conclude that \eqnok{constr} 
has a unique optimal solution $X_M^*$. Since $\|\K\|^2_F/2$ is the value of 
the objective function of \eqnok{constr} at $X=0$, we have
$\|\Lambda X_M^*-\K\|^2_F/2 \leq \|\K\|_F^2/2$, or equivalently
$\|\Lambda X_M^*\|_F^2 \leq 2 (\Lambda \K) \bu X_M^*$. Hence,
we have
\[
\left( \lambda_{\min}(\Lambda ) \right)^2 \, \|X_M^*\|_F^2 \le
\|\Lambda X_M^*\|_F^2 \leq 2 (\Lambda \K) \bu X_M^* \leq 2 \|X_M^*\|_F
\, \|\Lambda \K\|_F,
\]
which implies that
$\|X_M^*\|_F \leq 2 \|\Lambda \K\|_F/{\lambda^2_{\min}(\Lambda )}$.
Moreover, using the fact that
$\|X\|^2_F = \sum_{i=1}^m \sigma^2_i(X)$ for any
$X\in\Re^{p\times q}$, we easily see that
\beq \label{norm-singv}
\|X\|_F \le \sum_{i=1}^m \sigma_i(X).
\eeq
Since $X_M^*$ is feasible for \eqnok{constr},
it then follows from \eqnok{norm-singv} that $\|X_M^*\|_F \leq M$.
We have thus shown that inequality \eqnok{trX} holds.
\end{proof}

\vgap

\begin{lemma} \label{boundX}
For every $\lambda>0$, problem \eqnok{unconstr} has a unique optimal solution
$X_\lambda^*$. Moreover,
\beq \label{rX}
\|X_\lambda^*\|_F \le
\sum\limits_{i=1}^m \sigma_i(X_\lambda^*)
\le r_x := \min\left\{\frac{\|\K\|^2_F}{2\lambda},
\sum\limits_{i=1}^m \sigma_i(\Lambda ^{-1}\K)\right\}.
\eeq
\end{lemma}

\begin{proof}
As shown in Lemma \ref{constr-boundX}, the function $X \in \Re^{p \times q} \to
\|\Lambda  X - H \|_F^2$ is a (quadratic) strongly convex function.
Since the term $\lambda \sum_{i=1}^m \sigma_i(X)$ is convex in
$X$, it follows that the objective function of \eqnok{unconstr}
is strongly convex, from which we conclude that
\eqnok{unconstr} has a unique optimal solution $X_\lambda^*$.
Since $\|\K\|^2_F/2$ is the value of the objective function of
\eqnok{unconstr} at $X=0$, we have
\beq \label{singv-ineq1}
\lambda\sum\limits_{i=1}^m \sigma_i(X_\lambda^*) \le
\frac12\|\Lambda  X_\lambda^* - \K\|^2_F +
\lambda\sum\limits_{i=1}^m \sigma_i(X_\lambda^*)
\leq \frac12\|\K\|^2_F.
\eeq
Also, considering the objective function of \eqnok{unconstr} at
$X=\Lambda ^{-1}H$, we conclude that
\beq \label{singv-ineq2}
\lambda\sum\limits_{i=1}^m \sigma_i(X_\lambda^*) \le
\frac12\| \Lambda  X_\lambda^* - \K\|^2_F + \lambda\sum\limits_{i=1}^m \sigma_i(X_\lambda^*)
\leq \lambda \sum\limits_{i=1}^m
\sigma_i(\Lambda ^{-1}\K).
\eeq
Now, \eqnok{rX} follows immediately from \eqnok{norm-singv},
\eqnok{singv-ineq1} and \eqnok{singv-ineq2}.
\end{proof}

\vgap


\subsection{Cone programming reformulations}
\label{cone-reform}

In this subsection, we provide cone programming reformulations for 
problems \eqnok{constr} and \eqnok{unconstr}, respectively. 

\begin{proposition} \label{unconstr-conrepr}
Problem \eqnok{unconstr} can be reformulated as the following 
cone programming:
\beq \label{cone-repr1} \ba{cl}
\min\limits_{r, s, t, X, Y} & 2r + \lambda t\\
\mbox{s.t.}  &
\left(\ba{c} r+1  \\
r-1 \\
\vec(\Lambda X-\K)
\ea \right)
\in \cL^{pq +2}, \\ [20pt] 
& Y- \cG(X) +s I \succeq 0, \\ [6pt]
& ms + \tr(Y) - t \le  0, \ Y  \succeq  0, 
\ea 
\eeq
where $(r, s, t, X, Y) \in \Re \times \Re \times \Re \times
\Re^{p \times q} \times \Si{n}$ with $n:=p+q$
and $\cG(X)$ is defined in \eqnok{Gx}.
\end{proposition}

\begin{proof}
We first observe that \eqnok{unconstr} is equivalent to
\beq \label{lag_conic}
\ba{ll}
\min\limits_{r, X} & 2 r + \lambda t\\
\mbox{s.t.}
& \|\Lambda X - \K \|^2_{\F} \le 4r \\
& \sum\limits_{i=1}^m \sigma_i(X) - t \leq 0.
\ea 
\eeq
Using Lemma \ref{sum-sings} and the following relation
\beqas 4r \geq \|v\|^2 \Leftrightarrow 
\left(\ba{c} r+1  \\
r-1 \\
v \\ [4pt] 
\ea 
\right) \in \cL^{k+2},
\eeqas
for any $v \in \Re^k$ and $r \in \Re$, we easily see that
\eqnok{lag_conic} is equivalent to \eqnok{cone-repr1}
\end{proof}

\vgap

The following proposition can be similarly established.

\begin{proposition} \label{constr-conrepr}
Problem \eqnok{constr} can be reformulated as the following 
cone programming:
\beq \label{cone-repr2} \ba{cl}
\min\limits_{r, s, X, Y} & 2r \\
\mbox{s.t.}  & \left(\ba{c} r+1  \\
r-1 \\
\vec(\Lambda X-\K)
\ea \right)
\in \cL^{pq +2}, \\ [20pt] 
& Y- \cG(X) +s I \succeq 0, \\ [6pt]
& ms + \tr(Y) \le  M, \ \ Y  \succeq  0,
\ea 
\eeq
where $(r, s, X, Y) \in \Re \times \Re \times
\Re^{p \times q} \times \Si{n}$ with $n:=p+q$
and $\cG(X)$ is defined in \eqnok{Gx}.
\end{proposition}

\subsection{Smooth saddle point reformulations}
\label{sadd-reform}

In this section, we provide smooth saddle point reformulations for 
problems \eqnok{constr} and \eqnok{unconstr}.

\subsubsection{Smooth saddle point reformulations for \eqnok{unconstr}}

In this subsection, we reformulate \eqnok{unconstr} into a smooth saddle 
point problem that can be suitably solved by a variant of Nesterov's smooth 
method as described in Subsections \ref{nest-smooth} and \ref{sma-implement-2}. 

We start by introducing the following notation. For every $t \ge 0$,
we let $\Omega_t$ denote the set defined as
\beq \label{Omega}
\Omega_t := \{W\in \cS^{p+q}: 0 \preceq W \preceq tI/m, \tr(W) = t \}. 
\eeq

\begin{theorem} \label{sadpt-thm2}
For some $\eps \ge 0$, assume that $X_{\epsilon}$ is an $\epsilon$-optimal
solution of the smooth saddle point problem
\beq \label{unconstr-sadpt2} 
\min\limits_{X\in \cB_F^{p \times q}(r_x)} 
\max\limits_{W \in \Omega_1} \left\{\frac12\|\Lambda X - \K\|^2_F 
+ \lambda m \cG(X) \bu W \right\},
\eeq
where $\cG(X)$ and $r_x$ are defined in \eqnok{Gx} and \eqnok{rX},
respectively.  Then, $X_{\eps}$ is an $\eps$-optimal solution of 
problem \eqnok{unconstr}. 
\end{theorem}

\begin{proof}
This result follows immediately from Lemma \ref{boundX} and
relations \eqnok{char2} with $k=m$, \eqnok{unconstr} and
\eqnok{Omega} with $t=1$.
\end{proof}

\vgap

In addition to the saddle point (min-max) reformulation \eqnok{unconstr-sadpt2},
it is also possible to develop an alternative saddle point reformulation
based on the identity \eqnok{char1}. These two reformulations
can in turn be solved by a suitable method, namely Nesterov's smooth
approximation scheme \cite{Nest05-1}, for solving these min-max type problems, which
we will not describe in this paper. In our computational experiments,
we found that, among these two reformulations, the first one is
computationally superior than the later one.
Details of the computational comparison of these two approaches
can be found in the technical report (see \cite{LuMoYu08-1}),
which this paper originated from.

A more efficient method than the ones outlined in the previous paragraph
for solving \eqnok{unconstr} is based on solving the dual
of \eqnok{unconstr-sadpt2}, namely
the problem
\beq \label{unconstr-sadpt2d}
\max\limits_{W \in \Omega_1} \min\limits_{X\in \cB_F^{p \times q}(r_x)}
\left\{\frac12\|\Lambda X - \K\|^2_F
+ \lambda m \cG(X) \bu W \right\},
\eeq
whose objective function has the desirable property that it has
Lipschitz continuous gradient (see Subsection \ref{sma-implement-2}
for specific details).
In Subsections \ref{nest-smooth} and \ref{sma-implement-2},
we describe an algorithm, 
namely, a variant of Nesterov's smooth method, for solving 
\eqnok{unconstr-sadpt2d} which, as a by-product, yields a pair of primal and 
dual nearly-optimal solutions, and hence a nearly-optimal solution of 
\eqnok{unconstr-sadpt2}. Finally, Section \ref{comp} only reports computational 
results for the approach outlined in this paragraph since it is far superior
than the other two approaches outlined in the previous paragraph.

\subsubsection{Smooth saddle point reformulations for \eqnok{constr}}

In this subsection, we will provide a smooth saddle point reformulation 
for \eqnok{constr} that can be suitably solved by a variant of Nesterov's smooth method 
as described in Subsection \ref{nest-smooth}. 

By directly applying Theorem \ref{penalty-thm} to problem \eqnok{constr}, we 
obtain the following result. 

\begin{lemma} \label{penalty-lem}
Let $m:=\min(p,q)$. Suppose that $\bX\in\Re^{p\times q}$ satisfies
$\sum\limits_{i=1}^m \sigma_i(\bX) < M$ and let
$\gamma$ be a scalar such that $\gamma \ge \bar \gamma$, where 
$\bar \gamma$ is given by  
\beq \label{gamma}
\bar \gamma = \frac{\|\Lambda \bX-\K\|^2_{\F}/2}{M -\sum\limits_{i=1}^m 
\sigma_i(\bar{X})}.
\eeq
Then, the following statements hold:
\bi
\item[a)]
The optimal values of \eqnok{constr} and the penalized problem
\beq \label{exact-penalty1}
\min\limits_{X\in \Re^{p \times q}}
 \left\{\frac12\|\Lambda X - \K \|^2_{\F} +
\gamma \left[\sum_{i=1}^m \sigma_i(X)-M\right]^+ \right\}
\eeq
coincide, and the optimal solution solution $X_M^*$ of \eqnok{constr}
is an optimal solution of \eqnok{exact-penalty1};
\item[b)]
if $\epsilon \ge 0$ and $X_{\epsilon}$ is an $\epsilon$-optimal 
solution of problem \eqnok{exact-penalty1}, then the point
$X^{\epsilon}$ defined as
\beq \label{xeps-theta}
X^{\epsilon} :=\frac{X_{\epsilon} + \theta\bar{X}}{1 + \theta}, \ \ \ \ \
\mbox{where} \ \theta :=
\frac{\left[\sum\limits_{i=1}^m \sigma_i(X_{\epsilon}) - M \right]^+}
{M -\sum\limits_{i=1}^m \sigma_i(\bX)},
\eeq
is an $\epsilon$-optimal solution of \eqnok{constr}.
\ei
\end{lemma}

\vgap
We next provide a smooth saddle point reformulation for problem \eqnok{constr}. 

\begin{theorem} \label{exactpenal-thm2}
Let $m:=\min(p,q)$. Suppose that $\bX\in\Re^{p\times q}$ satisfies
$\sum\limits_{i=1}^m \sigma_i(\bX) < M$
and let $\gamma$ be a scalar such that $\gamma \ge \bar \gamma$, 
where $\bar\gamma$ is defined in \eqnok{gamma}. For some $\epsilon \geq 0$,
assume that $X_{\epsilon}$ is an $\epsilon$-optimal solution
of the problem
\beq \label{constr-sadpt2} 
\min\limits_{X\in \cB_F^{p \times q}(\trx)} 
\max\limits_{(t,W)\in \tOmega}\left\{\frac12\|\Lambda X - \K \|^2_{\F} + 
\gamma (m \cG(X) \bu W -Mt) \right\},
\eeq
where $\tilde r_x$ is defined in \eqnok{trX} and $\tOmega$ is defined as
\beq \label{tOmega}
\tOmega := \{(t,W)\in \Re \times  \cS^{p+q}: W \in \Omega_t,
\, 0 \le t \le 1 \}.
\eeq
Let $X^{\epsilon}$ be defined in \eqnok{xeps-theta}. Then, $X^{\epsilon}$ 
is an $\epsilon$-optimal solution of \eqnok{constr}.
\end{theorem}
 
\begin{proof}
Let $X_M^*$ denote the unique optimal solution of \eqnok{constr}.
Then, $X_M^*$ is also an optimal solution of  \eqnok{exact-penalty1}
in view of Lemma \ref{penalty-lem}(a), and
$X_M^*$ satisfies $X_M^* \in \cB_F^{p \times q}(\trx)$
due to Lemma \ref{constr-boundX}.  
Also, relation \eqnok{char4} with $\alpha=1$, $\beta=-M$ and $k=m$
implies that the objective functions of problems 
\eqnok{exact-penalty1} and \eqnok{constr-sadpt2} are equal to each other 
over the whole space $\Re^{p \times q}$. The above observations then
imply that $X_M^*$ is also an optimal solution of \eqnok{constr-sadpt2}
and that problems \eqnok{exact-penalty1} and \eqnok{constr-sadpt2} have
the same optimal value. 
Since by assumption $X_{\epsilon}$ is an $\epsilon$-optimal
solution of \eqnok{constr-sadpt2}, it follows that $X_{\epsilon}$
is also an $\epsilon$-optimal solution of problem \eqnok{exact-penalty1}.
The latter conclusion together with Lemma \ref{penalty-lem}(b)
immediately yields the conclusion of the theorem.
\end{proof}

\vgap

The saddle point (min-max) reformulation \eqnok{constr-sadpt2} can be solved by 
a suitable method, namely, Nesterov's smooth approximation scheme \cite{Nest05-1}, 
which we will not describe in this paper. A more efficient method for solving 
\eqnok{constr} is based on solving the dual
of \eqnok{constr-sadpt2}, namely
the problem
\beq \label{constr-sadpt2d}
\max\limits_{(t,W)\in \tOmega}\min\limits_{X\in \cB_F^{p \times q}(\trx)}
\left\{\frac12\|\Lambda X - \K \|^2_{\F} + \gamma (m\cG(X) \bu W -Mt) \right\},
\eeq
whose objective function has the desirable property that it has Lipschitz 
continuous gradient (see Subsection \ref{sma-implement-3} for specific details). 
In Subsections \ref{nest-smooth} and \ref{sma-implement-3}, we describe an 
algorithm, namely a variant of Nesterov's smooth method, for solving 
\eqnok{constr-sadpt2d} which, as a by-product, yields a pair of primal and 
dual nearly-optimal solutions, and hence a nearly-optimal solution of 
\eqnok{constr-sadpt2}. 

\section{Numerical methods}
\label{numerical}

In this section, we discuss numerical methods for solving problem \eqnok{unconstr}. More 
specifically, Subsection \ref{nest-smooth} reviews a variant of Nesterov's smooth method 
\cite{tseng08}, for solving a convex minimization 
problem over a relatively simple set with a smooth objective function that 
has Lipschitz continuous gradient. In Subsections \ref{sma-implement-2} and 
\ref{sma-implement-3}, we present the implementation details of the variant of 
Nesterov's smooth methd for solving the reformulations \eqnok{unconstr-sadpt2d} 
of problem \eqnok{unconstr} and \eqnok{constr-sadpt2d} of problem \eqnok{constr}, 
respectively.


The implementation details of the other formulations discussed in the paper, more 
specifically, the reformulations \eqnok{unconstr-sadpt2} of problem \eqnok{unconstr} and 
\eqnok{constr-sadpt2} of problem \eqnok{constr} will not be presented here. 
The implementation details of some other reformulations of problems \eqnok{unconstr} 
and \eqnok{constr} can be found in Subsection $4.2$ of \cite{LuMoYu08-1}.

%

\subsection{Review of a variant of Nesterov's smooth method}
\label{nest-smooth}

In this subsection, we review a variant of Nesterov's smooth 
first-order method \cite{Nest83-1,Nest05-1} that is proposed 
by Tseng \cite{tseng08} for solving a class of smooth convex 
programming (CP) problems.

Let $\cU$ and $\cV$ be normed vector spaces with the respective
norms denoted by $\|\cdot\|_\cU$ and $\|\cdot\|_\cV$.
We will discuss a variant of Nesterov's smooth first-order 
method for solving the class of CP problems
\beq \label{convex-opt}
\min\limits_{u\in U} f(u)
\eeq  
where the objective function $f:U \to \Re$ has the form
\beq \label{fu}
f(u) := \max\limits_{v\in V}\phi(u,v), \ \ \forall u \in U,
\eeq
for some continuous function $\phi:U \times V \to \Re$ and
nonempty compact convex subsets $U \subseteq \cU$ and $V \subseteq \cV$.
We make the following assumptions regarding the function $\phi$:

\vgap

{\bf B.1} for every $u \in U$, the function
$\phi(u,\cdot): V \to \Re $ is {\it strictly} concave;

{\bf B.2} for every $v \in V$, the function
$\phi(\cdot,v): U \to \Re $ is convex differentiable;

{\bf B.3} the function $f$ is 
$L$-Lipschitz-differentiable on $U$ with respect to $\|\cdot\|_\cU$ (see \eqnok{sete}).

\vgap

It is well-known that Assumptions B.1 and B.2 imply that the function $f$
is convex differentiable, and that its gradient is given by
\beq \label{gu}
\nabla f(u) = \nabla_u \phi(u,v(u)), \ \ \forall u\in U,
\eeq
where $v(u)$ denotes the unique solution of \eqnok{fu}
(see for example Proposition B.25 of \cite{Bert99}). Moreover, problem
\eqnok{convex-opt} and its dual, namely:
\beq \label{dual-prob}
\max\limits_{v \in V} \, \{ g(v):= \min\limits_{u\in U} \phi(u, v) \},
\eeq
both have optimal solutions $u^*$ and $v^*$ such that $f(u^*)=g(v^*)$.
Finally, using Assumption B.3, Lu \cite{Lu07-1} recently showed that problem 
\eqnok{convex-opt}-\eqnok{fu} and its dual problem \eqnok{dual-prob} 
can be suitably solved by Nesterov's smooth method \cite{Nest05-1}, 
simultaneously. We shall notice, however, that Nesterov's smooth method 
\cite{Nest05-1} requires solving two prox-type subproblems per iteration. 
More recently, Tseng \cite{tseng08} proposed a variant of Nesterov's smooth 
method described as follows, which needs to solve one prox subproblem per 
iteration only.   

Let $\dU:U \to \Re$ be a differentiable strongly convex function
with modulus $\usigma >0$ with respect to $\|\cdot\|_\cU$, i.e.,
\beq \label{strong_h}
\dU(u) \ge \dU(\tu) + \langle \nabla\dU(\tu), u-\tu \rangle +
\frac{\usigma}{2} \|u-\tu\|_\cU^2, \ \ \forall u,\tu \in \setU.
\eeq
Let $u_0$ be defined as 
\beq \label{u0}
u_0 = \arg\min\{\dU(u): \ u \in U\}.
\eeq
By subtracting the constant $\dU(u_0)$ from the function $\dU(\cdot)$,
we may assume without any loss of generality that $\dU(u_0)=0$.
The Bregman distance $d_\uh : U \times U \to \Re$
associated with $\uh$ is defined as
\beq \label{Bregdist}
d_\uh(u;\tu)=\uh(u)-l_\uh(u;\tu), \ \ \forall u,\tu \in \setU,
\eeq
where $l_\uh: \cU \times U \to \Re$ is the ``linear approximation''
of $\uh$ defined as
\[
l_\uh(u;\tu) = \uh(\tu)+\langle \nabla p_U(\tu), u-\tu \rangle,
\ \ \forall (u,\tu) \in \cU \times U.
\]
Similarly, we can define the function $l_f(\cdot;\cdot)$ that will be 
used subsequently. 

We now describe the variant of Nesterov's smooth method proposed by Tseng 
\cite{tseng08} for solving problem \eqnok{convex-opt}-\eqnok{fu} and its 
dual problem \eqnok{dual-prob}. It uses a sequence $\{\alpha_k\}_{k\ge 0}$ 
of scalars satisfying the following condition:
\beq \label{condseq}
0 < \alpha_k \le \left(\sum_{i=0}^k \alpha_i \right)^{1/2}, \ \forall k \ge 0.
\eeq
Clearly, \eqnok{condseq} implies that $\alpha_0 \in (0,1]$.

\gap

\noindent
\begin{minipage}[h]{6.6 in}
{\bf Variant of Nesterov's smooth algorithm:} \\ [5pt]
Let $u_0\in U$ and $\{\alpha_k\}_{k\ge 0}$ satisfy 
\eqnok{u0} and \eqnok{condseq}, respectively. \\ [4pt]
Set $u^{sd}_0=u_0$, $v_0 = 0 \in \cV$, $\tau_0=1$ and $k=1$;
\begin{itemize}
\item[1)]
Compute $v(u_{k-1})$ and $\nabla f(u_{k-1})$.
\item[2)]
Compute $(u^{sd}_{k},u^{ag}_{k}) \in U \times U$ and
$v_k \in V$ as
\beqa
v_k &\equiv& (1-\tau_{k-1})v_{k-1} + \tau_{k-1}v(u_{k-1}) \nn \\
u^{ag}_{k} &\equiv& \argmin \left \{ \frac{L}{\usigma} \, d_\uh(u;u_0)
+ \sum_{i=0}^{k-1} \alpha_i \, l_f(u;u_i) : u \in \setU \right \} \label{proxsub1} \\
u^{sd}_k &\equiv& (1-\tau_{k-1}) u^{sd}_{k-1} + \tau_{k-1} u^{ag}_k. \nn
\eeqa
\item[3)]
Set $\tau_k = \alpha_k/(\sum_{i=0}^k \alpha_i)$ and
$u_k = (1-\tau_k) u^{sd}_{k} + \tau_k u^{ag}_k$.
\item[4)]
Set $k \leftarrow k+1$ and go to step 1). 
\end{itemize}
\noindent
{\bf end}
\end{minipage}
\vgap

We now state the main convergence result regarding the variant of Nesterov's
smooth algorithm for solving problem \eqnok{convex-opt}and its dual \eqnok{dual-prob}. 
Its proof is given in Corollary 3 of Tseng \cite{tseng08}.

\begin{theorem} \label{mtm-sm}
The sequence $\{(u^{sd}_k, v_k)\} \subseteq U \times V$ generated
by the variant of Nesterov's smooth algorithm satisfies
\beq \label{gap}
0 \le f(u^{sd}_k) - g(v_k) \le
\frac{L\DU}{\usigma ( \sum_{i=0}^{k-1} \alpha_i)}, \ \ \forall k \ge 1,
\eeq
where
\beq \label{D}
\DU = \max \{\dU(u): \ u \in U\}.
\eeq
\end{theorem}

A typical sequence $\{\alpha_k\}$ satisfying \eqnok{condseq} is the
one in which $\alpha_k=(k+1)/2$ for all $k \ge 0$. With this
choice for $\{\alpha_k\}$, we have the following specialization
of Theorem \ref{mtm-sm}.

\begin{corollary} \label{mtm-sm1}
If $\alpha_k=(k+1)/2$ for every $k \ge 0$, then the
sequence $\{(u^{sd}_k, v_k)\} \subseteq U \times V$ generated
by the variant of Nesterov's smooth algorithm satisfies
\[
0 \le f(u^{sd}_k) - g(v_k) \le
\frac{4L\DU}{\usigma k(k+1)}, \ \ \forall k \ge 1,
\]
where $\DU$ is defined in $\eqnok{D}$. 
Thus, the iteration-complexity of finding an $\epsilon$-optimal
solution to \eqnok{convex-opt} and its dual \eqnok{dual-prob}
by the variant of Nesterov's smooth algorithm does not exceed
$2[(L\DU)/(\usigma \epsilon)]^{1/2}$.
\end{corollary}

Before ending this subsection, we state sufficient conditions for the
function $\phi$ to satisfy Assumptions B.1-B.3.
The proof of the following result can be found in Theorem 1 of
\cite{Nest05-1}.

\begin{proposition} \label{Nest}
Let a norm $\|\cdot\|_{\cal V}$ on $\cV$ be given. 
Assume that $\phi:U \times V \to \Re$ has the form
\beq \label{phiuv}
\phi(u,v) = \theta(u) + \langle u , {\cal E}v \rangle - h(v), \ \ \
\forall (u,v) \in U \times V,
\eeq
where ${\cal E}: {\cal V} \to {\cal U}^*$ is a linear map,
$\theta:U \to \Re$ is $L_\theta$-Lipschitz-differentiable in
$U$ with respect to $\|\cdot\|_{\cU}$, and
$h: V \to \Re$ is a differentiable strongly convex function with modulus
$\sigma_V>0$ with respect to $\|\cdot\|_{\cal V}$.
Then, the function $f$ defined by $\eqnok{fu}$ is
$(L_\theta+\|{\cal E}\|_{\cU,\cV}^2/\sigma_V)$-Lipschitz-differentiable in
$U$ with respect to $\|\cdot\|_{\cU}$. As a consequence,
$\phi$ satisfies Assumptions B.1-B.3 with norm $\|\cdot\|_\cU$ and
$L= L_\theta+\|{\cal E}\|_{\cU,\cV}^2/\sigma_V$.
\end{proposition}

We will see in Section \ref{numerical} that all
saddle-point reformulations \eqnok{convex-opt}-\eqnok{fu}
of problems \eqnok{constr} and \eqnok{unconstr} studied
in this paper have the property that the corresponding
function $\phi$ can be expressed as in \eqnok{phiuv}.

\subsection{Implementation details of the variant of Nesterov's smooth method 
for \eqnok{unconstr-sadpt2d}}
\label{sma-implement-2}



The implementation details of the variant of Nesterov's smooth method 
(see Subsection \ref{nest-smooth}) for solving formulation
\eqnok{unconstr-sadpt2d} (that is, the dual of \eqnok{unconstr-sadpt2})
are addressed in this subsection. In particular, we describe in the
context of this formulation the prox-function, the Lipschitz constant $L$
and the subproblem \eqnok{proxsub1} used by
the variant of Nesterov's smooth algorithm of Subsection \ref{nest-smooth}.

For the purpose of our implementation, we reformulate problem 
\eqnok{unconstr-sadpt2d} into the problem
\beq \label{unconstr-sadpt2d-1}
\min\limits_{W \in \Omega_1}\max\limits_{X\in \cB_F^{p \times q}(1)} 
\left\{-\lambda m r_x\cG(X) \bu W - \frac12\|r_x \Lambda X - \K\|^2_F\right\}
\eeq 
obtained by scaling the variable $X$ of \eqnok{unconstr-sadpt2d} as
$X \leftarrow X/r_x$, and multiplying the resulting formulation by $-1$.
 From now on, we will focus on formulation \eqnok{unconstr-sadpt2d-1}
rather than \eqnok{unconstr-sadpt2d}. 

Let $n:= p+q$, $u := W$, $v :=X$ and define
\beqas
&& U := \Omega_1 \subseteq \cS^n =: \cU, \\ [4pt] 
&& V := \cB_F^{p \times q}(1) \subseteq \Re^{p\times q}  =: \cV,
\eeqas
and
\beq \label{phiuv1}
\phi(u,v) := -\lambda m r_x\cG(v) \bu u - \frac12\|r_x \Lambda v - \K\|^2_F, \ \ \forall 
(u,v)\in U \times V,
\eeq
where $\Omega_1$ is defined in \eqnok{Omega}. Also, assume that the
norm on $\cU$ is chosen as
\[
\|u\|_\cU := \|u\|_F, \ \ \forall u\in\cU.
\]
Our aim now is to show that $\phi$ satisfies Assumptions B.1-B.3
with $\|\cdot\|_\cU$ as above and some Lipschitz constant
$L>0$, and hence that the variant of Nesterov's method can be applied to
the corresponding saddle-point formulation \eqnok{unconstr-sadpt2d-1}.
This will be done with the help of Proposition \ref{Nest}. Indeed, the
function $\phi$ is of the form \eqnok{phiuv} with $\theta \equiv 0$ and
the functions $\cE$ and $h$ given by
\beqas
\cE v &:=& -\lambda m r_x\cG(v), \ \ \forall v \in \cV, \\ 
h(v) &:=& \frac12\|r_x \Lambda v - \K\|^2_F, \ \ \forall v\in \cV. 
\eeqas
Assume that we fix the norm on $\cV$ to be the Frobenius norm, i.e.,
$\|\cdot\|_\cV=\|\cdot\|_F$.
Then, it is easy to verify that the above function $h$ is
strongly convex with modulus $\sigma_V:=r_x^2/\|\Lambda^{-1}\|^2$
with respect to $\|\cdot\|_\cV=\|\cdot\|_F$.
Now, using \eqnok{normC1}, we obtain
\beqa
\|\cE\|_{\cU,\cV} &=& \max \left\{ \|\lambda m r_x\cG(v)\|^*_\cU: 
\ v \in \cV, \|v\|_\cV \le 1 \right\},   
\nn \\ [5pt]
&=& \lambda m r_x  \max\left\{ \|\cG(v)\|_F:\ v\in \cV,
\|v\|_F \le 1 \right\}, \nn \\ [5pt]
&=& \lambda m r_x  \max\left\{\sqrt{2}\|v\|_F:\ v\in \cV, 
\|v\|_F \le 1\right\} = \sqrt2\lambda m r_x.
 \label{cE-2}
\eeqa 
Hence, by Proposition \ref{Nest}, we conclude that
$\phi$ satisfies Assumptions B.1-B.3 with
$\|\cdot\|_\cU=\|\cdot\|_F$ and
\[
L = \|\cE\|^2_{U,V}/{\vsigma} = 2\lambda^2m^2\|\Lambda ^{-1}\|^2. 
\] 

The prox-function $\dU(\cdot)$ for the set $U$ used in the
variant of Nesterov's algorithm is defined as
\beq \label{dv3}
\dU(u) = \tr(u\log u) + \log n, \ \ \forall u \in U=\Omega_1.
\eeq
We can easily see that $\dU(\cdot)$ is a strongly differentiable convex 
function on $U$ with modulus $\usigma = m$ with respect to the norm 
$\|\cdot\|_\cU=\|\cdot\|_F$. Also, it is easy to verify that
$\min\{p_U(u) : u \in U\} = 0$ and that
\beqa
u_0 & := & \arg\min_{u \in U} p_U(u) = I/n, \label{u00} \\
\DU & := & \max_{u \in U} \dU(u) = \log (n/m). \nn
\eeqa

As a consequence of the above discussion and Theorem \ref{mtm-sm1},
we obtain the following result.

\begin{theorem} \label{compl-3-1}
For a given $\epsilon>0$, the variant of Nesterov's smooth method applied to 
\eqnok{unconstr-sadpt2d} finds an $\epsilon$-optimal solution of problem 
\eqnok{unconstr-sadpt2d} and its dual, and hence of problem \eqnok{unconstr},
in a number of iterations which does not exceed
\beq \label{sma-compl3}
\left\lceil
\frac{2\sqrt{2}\lambda \|\Lambda ^{-1}\|}{\sqrt{\epsilon}} \sqrt{m\log (n/m)}
\right\rceil.
\eeq     
\end{theorem}

\gap


We observe that the iteration-complexity given in \eqnok{sma-compl3} is in 
terms of the transformed data of problem \eqnok{multi-reg4}. We 
next relate it to the original data of problem \eqnok{multi-reg4}.

\begin{corollary} \label{compl-3-2}
For a given $\epsilon>0$, the variant of Nesterov's smooth method applied to 
\eqnok{unconstr-sadpt2d} finds an $\epsilon$-optimal solution of problem 
\eqnok{unconstr-sadpt2d} and its dual, and hence of problem \eqnok{unconstr},
in a number of iterations which does not exceed
\[
\left\lceil
\frac{2\sqrt{2}\lambda\|(A^TA)^{-1/2}\|}{\sqrt{\epsilon}} \sqrt{m\log (n/m)}
\right\rceil.
\]  
\end{corollary}

\begin{proof}
We know from Subsection \ref{reduction} that $A^TA=Q\Lambda^2 Q^T$,
where $Q\in \Re^{p\times p}$ is an orthonormal matrix. Using this relation, 
we have 
\[
\|\Lambda^{-1}\| = \|\Lambda^{-2}\|^{1/2} = \|(A^TA)^{-1}\|^{1/2}=\|(A^TA)^{-1/2}\|.
\]
The conclusion immediately follows from this identity and 
Theorem \ref{compl-3-1}.
\end{proof}

\gap

It is interesting to note that the iteration-complexity of
Corollary \ref{compl-3-2} depends on the data matrix $A$ but not on $B$.
Based on the discussion below,
the arithmetic operation cost per iteration of the variant of
Nesterov's smooth method when applied to problem
\eqnok{constr-sadpt2d} is bounded by ${\cal O}(mpq)$ where
$m=\min(p,q)$, due to the fact that its most expensive operation consists of
finding a partial singular value decomposition of a $p \times q$ matrix $h$
as in \eqnok{singdec}.
Thus, the overall arithmetic-complexity of the variant of
Nesterov's smooth method when applied to \eqnok{unconstr-sadpt2d} is 
\[
{\cal O}\left(\frac{\lambda\|(A^TA)^{-1/2}\|}{\sqrt{\epsilon}}m^{3/2}
pq\sqrt{\log (n/m)}\right).
\] 

After having completely specified all the ingredients required by the variant 
of Nesterov's smooth method for solving \eqnok{unconstr-sadpt2d-1}, we now
discuss some of the computational technicalities involved in the
actual implementation of the method.

First, recall that, for a given $u\in U$, the optimal solution
for the maximization subproblem \eqnok{fu} needs to be found
in order to compute the gradient of $\nabla f(u)$. 
Using \eqnok{phiuv1} and the fact that $V=\cB_F^{p \times q}(1)$,
we see that the maximization problem \eqnok{fu} is equivalent to
\beq \label{qp-ball}
\min\limits_{v \in\cB_F^{p \times q}(1)} 
\frac12\|r_x \Lambda v - \K\|^2_F + G \bullet v,
\eeq
where $G:={\cal G}^*(u) \in\Re^{p\times q}$.
We now briefly discuss how to solve \eqnok{qp-ball}.
For any $\xi \ge 0$, let
\[
v(\xi) = (r^2_x\Lambda ^2+\xi I)^{-1} (r_x \Lambda H-G), \ \ \ \ 
\Psi(\xi) = \|v(\xi)\|^2_F - 1.
\]
If $\Psi(0) \le 0$, then clearly $v(0)$ is the optimal solution of
problem \eqnok{qp-ball}. Otherwise, the optimal solution of problem
\eqnok{qp-ball} is equal to $v(\xi^*)$, 
where $\xi^*$ is the root of the equation $\Psi(\xi)=0$. The latter
can be found by well-known root finding schemes specially taylored
for solving the above equation.

In addition, each iteration of the variant of Nesterov's smooth method
requires solving subproblem \eqnok{proxsub1}. In view of \eqnok{gu} and
\eqnok{phiuv1}, it is easy to see that
for every $u\in U$, we have $\nabla f(u) = \cG(v)$ for some 
$v\in\Re^{p \times q}$.
Also, $\nabla \dU(u_0) = (1-\log n) I$ due to
\eqnok{dv3} and \eqnok{u00}. These remarks together 
with \eqnok{Bregdist} and \eqnok{dv3} imply that
subproblem \eqnok{proxsub1} is of the form
\beq \label{fg-subprob}
\min\limits_{u\in \Omega_1} \ \left(\varsigma  I + \cG(h)\right) \bullet u
+ \tr(u \log u)
\eeq
for some real scalar $\varsigma$ and $h\in\Re^{p \times q}$, where $\Omega_1$
is given by \eqnok{Omega}.

We now present an efficient approach for solving \eqnok{fg-subprob}
which, instead of finding the eigenvalue factorization of
the $(p+q)$-square matrix $\varsigma  I + \cG(h)$, computes
the singular value decomposition of the smaller $p \times q$-matrix $h$.
First, we compute a singular value decomposition of $h$, i.e.,
$h = \tU \Sigma \tV^T$, where $\tU \in \Re^{p \times m}$,
$\tV \in \Re^{q \times m}$ and $\Sigma$ are such that
\beq \label{singdec}
\tU^T \tU = I, \ \ \ \Sigma = \diag(\sigma_1(h), \ldots, \sigma_m(h)), \ \ \
\tV^T \tV = I,
\eeq
where $\sigma_1(h), \ldots, \sigma_m(h)$ are the $m=\min(p,q)$ singular
values of $h$.
Let $\xi_i$ and $\eta_i$ denote the $i$th column of $\tU$ and $\tV$, 
respectively. Using \eqnok{Gx}, it is easy to see that 
\beq \label{fi}
f^i = \frac{1}{\sqrt{2}}\left(\ba{c}\eta_i \\ \xi_i \ea \right), \ i=1,\ldots, 
m; \ \ \ f^{m+i} = \frac{1}{\sqrt{2}}\left(\ba{c}\eta_i \\ -\xi_i \ea \right), 
\ i=1,\ldots, m,
\eeq
are orthonormal eigenvectors of $\cG(h)$ with eigenvalues 
$\sigma_1(h), \ldots, \sigma_m(h), -\sigma_1(h), \ldots, -\sigma_m(h)$,
respectively. 
Now let $f^i \in \Re^{n}$ for $i=2m+1, \ldots, n$ be such that
the matrix  $F:=(f^1, f^2, \ldots, f^{n})$ satisfies $F^TF = I$.
It is well-known that the vectors $f^i \in \Re^{n}$,
$i=2m+1, \ldots, n$, are eigenvectors of $\cG(h)$ corresponding to the zero
eigenvalue (e.g., see \cite{BenNem01}). Thus, we obtain the following
eigenvalue decomposition of $\varsigma I + \cG(h)$:
\[
\varsigma I + \cG(h) = F \diag(a) F^T,  \ \ \ a=\varsigma e + (\sigma_1(h), \ldots, 
\sigma_m(h), -\sigma_1(h), \ldots, -\sigma_m(h), 0, \ldots, 0)^T.
\]  
Using this relation and \eqnok{Omega} with $t=1$,
it is easy to see that the optimal 
solution of \eqnok{fg-subprob} is $v^* = F\diag(w^*)F^T$, where 
$w^*\in \Re^{n}$ is the unique optimal solution of the problem
\beq \label{reduced-prob}
\ba{ll}
\min & a^T w + w^T \log w \\
\mbox{s.t.} & e^T w = 1, \\ 
& 0 \le w \le e/m.             
\ea
\eeq
It can be easily shown that $w^*_i = \min\{\exp(-a_i-1-\xi^*),1/m\}$, 
where $\xi^*$ is the unique root of the equation
\[
\sum^{n}_{i=1} \min\{\exp(-a_i-1-\xi),1/m\} -1 = 0.
\]
Let $\vartheta := \min\{\exp(-\varsigma-1-\xi^*),1/m\}$. In view of the above formulas for 
$a$ and $w^*$ , we immediately see that  
\beq \label{wi}
w^*_{2m+1}=w^*_{2m+2}= \cdots = w^*_{n} = \vartheta.
\eeq
Further, using the fact that $FF^T = I$, we have 
\[
\sum\limits_{i=2m+1}^{n} f^i(f^i)^T = I - \sum\limits_{i=1}^{2m}f^i(f^i)^T. 
\]
Using this result and \eqnok{wi}, we see that the optimal solution 
$v^*$ of \eqnok{fg-subprob} can be efficiently computed as
\[
v^* = F \diag(w^*) F^T = \sum\limits_{i=1}^{n} w^*_i f^i(f^i)^T 
 =  \vartheta I +  \sum\limits_{i=1}^{2m} (w^*_i-\vartheta)f^i (f^i)^T,     
\] 
where the scalar $\vartheta$ is defined above and the vectors
$\{f^i: i=1,\ldots 2m\}$ are given by \eqnok{fi}.


Finally, to terminate the variant of Nesterov's smooth method, we need 
to evaluate the primal and dual objective functions of problem 
\eqnok{unconstr-sadpt2d-1}. As mentioned above, the 
primal objective function $f(u)$ of \eqnok{unconstr-sadpt2d-1} can be 
computed by solving a problem of the form \eqnok{qp-ball}. Additionally, 
in view of \eqnok{char2} and \eqnok{Omega}, the dual objective function $g(v)$ 
of \eqnok{unconstr-sadpt2d-1} can be computed as
\[
g(v) = -\frac12\|r_x \Lambda v - \K\|^2_F 
- \lambda r_x \sum\limits_{i=1}^m \sigma_i(v), \ \ \forall v \in V.
\]

\subsection{Implementation details of the variant of Nesterov's smooth method 
for \eqnok{constr-sadpt2d}}
\label{sma-implement-3}

The implementation details of the variant of Nesterov's smooth method
(see Subsection \ref{nest-smooth}) for solving formulation
\eqnok{constr-sadpt2d} (that is, the dual of \eqnok{constr-sadpt2})
are addressed in this subsection. In particular, we describe in the
context of this formulation the prox-function, the Lipschitz constant $L$
and the subproblem \eqnok{proxsub1} used by
the variant of Nesterov's smooth algorithm of Subsection \ref{nest-smooth}.

For the purpose of our implementation, we reformulate problem 
\eqnok{constr-sadpt2d} into the problem
\beq \label{constr-sadpt2d-1}
\min\limits_{(t,W)\in \tOmega}\max\limits_{X\in \cB_F^{p \times q}(1)}
\left\{-\gamma [m\trx\cG(X) \bu W -Mt] - \frac12\|\trx\Lambda X - \K \|^2_{\F}\right\}
\eeq
obtained by scaling the variables $X$ of \eqnok{constr-sadpt2d} as
$X \leftarrow X/\trx$, and multiplying the resulting formulation by $-1$.
 From now on, our discussion in this subsection will focus on formulation
\eqnok{constr-sadpt2d-1} rather than \eqnok{constr-sadpt2d}. 
 
Let $n:= p+q$, $u :=(t,W)$, $v :=X$ and define
\beqas
&& U := \tOmega \subseteq \Re \times \cS^{n} =: \cU, \\ [4pt]
&& V :=  \cB_F^{p \times q}(1) \subseteq \Re^{p \times q} =: \cV
\eeqas
and
\beq \label{phiuv2}
\phi(u,v) := -\gamma [m\trx\cG(v) \bu W -Mt] -
\frac12\|\trx\Lambda v - \K \|^2_{\F}, \ \ \forall (u,v)\in U \times V,
\eeq
where $\tOmega$ is defined in \eqnok{tOmega}.
Also, assume that the norm on $\cU$ is chosen as
\[
\|u\|_\cU := (\xi t^2 + \|W\|^2_F)^{1/2}, \ \ \forall u=(t,W) \in \cU,
\]
where $\xi$ is a positive scalar that will be specified later.
Our aim now is to show that $\phi$ satisfies Assumptions B.1-B.3
with $\|\cdot\|_\cU$ as above and some Lipschitz constant
$L>0$, and hence that the variant of Nesterov's method can be applied to
the corresponding saddle-point formulation \eqnok{constr-sadpt2d-1}.
This will be done with the help of Proposition \ref{Nest}. Indeed, the
function $\phi$ is of the form \eqnok{phiuv} with $\theta$,
$\cE$ and $h$ given by
\beqa
\theta(u) &:=& \gamma Mt, \ \ \forall u=(t,W) \in \cU, \nn \\
\cE v &:=& (0,-\gamma m \tilde r_x\cG(v)), \ \ \forall v \in \cV, \label{E10} \\
h(v) &:=& \frac12\|\tilde r_x \Lambda v - \K\|^2_F, \ \ \forall v\in \cV. \nn
\eeqa
Clearly, $\theta$ is a linear function, and thus it is a
$0$-Lipschitz-differentiable function on $U$ with respect to $\|\cdot\|_\cU$.
Now, assume that we fix the norm on $\cV$ to be the Frobenius norm, i.e.,
$\|\cdot\|_\cV=\|\cdot\|_F$.
Then, it is easy to verify that the above function $h$ is
strongly convex with modulus $\sigma_V:=\tilde r_x^2/\|\Lambda^{-1}\|^2$
with respect to $\|\cdot\|_\cV=\|\cdot\|_F$.
Now, using \eqnok{normC1}, \eqnok{E10} and the fact that
\beq \label{dnormU-1}
\|u\|^*_{\cU} = (\xi^{-1} t^2 + \|W\|^2_F)^{1/2}, \ 
\forall u=(t,W) \in \cU^*=\cU,
\eeq
we obtain
\beqa
\|\cE\|_{\cU,\cV} &=& \max \left\{ \|(0,-\gamma m \trx\cG(v))\|^*_\cU: 
\ v \in \cV, \|v\|_\cV \le 1 \right\},   
\nn \\ [5pt]
&=& \gamma m \trx \max\left\{ \|\cG(v)\|_F:\ v\in \cV,
\|v\|_F \le 1 \right\}, \nn \\ [5pt]
&=& \gamma m \trx \max\left\{\sqrt{2}\|v\|_F:\ v\in \cV, 
\|v\|_F \le 1\right\} = \sqrt2\gamma m \trx .
\eeqa 
Hence, by Proposition \ref{Nest}, we conclude that
$\phi$ satisfies Assumptions B.1-B.3 with
$\|\cdot\|_\cU=\|\cdot\|_F$ and
\beq \label{lipconst}
L = L_\theta+ \|\cE\|^2_{U,V}/{\vsigma} = 2\gamma^2m^2\|\Lambda ^{-1}\|^2.
\eeq

We will now specify the prox-function $\dU$ for the set $U$
used in the variant of Nesterov's algorithm. We let
\beq
\dU(u) = \tr(W\log W) + a t \log t + b t + c,
\ \ \forall u=(t,W) \in U, \label{du-1}
\eeq
where
\beq \label{abc}
a:=\log \frac{n}{m}, \ \ \ b := \log n - a -1 = \log m - 1, \ \ \
c:= a + 1.
\eeq
For a fixed $t \in [0,1]$, it is easy to see that
\[
\min_{W \in \Omega_t} \dU(t,W) = \psi(t)
:= t \log\frac{t}{n} + a t \log t + bt +c,
\]
and that the minimum is achieved at $W=tI/n$. Now,
\[
\psi'(1)= \log\frac{1}{n} + 1 + a ( \log 1 + 1) + b = 1 - \log n + a + b =0,
\]
where the last equality follows from the second identity in \eqnok{abc}.
These observations together with \eqnok{tOmega} allow us to conclude that
\beqa
&\arg\min_{u \in U} \dU(u) = u_0 := (1, I/n), \label{u0-2} \\
&\min_{u\in U} \dU(u) = \psi(1) = -\log n + b + c = 0,  \label{DU}
\eeqa
where the last equality is due to second and third identities in \eqnok{abc}.
Moreover, it is easy to see that
\beq \label{DU1}
\DU := \max_{u\in U} \dU(u) = \max_{t \in [0,1]}
t \log \frac{t}{m} + a t \log t + b t + c
= c + \max \left\{ 0, b - \log m \right\} = 1+ \log \frac{n}m,
\eeq
where the last identity is due to \eqnok{abc}. Also,
we easily see that $\dU(\cdot)$ is a strongly differentiable
convex function on $U$ with modulus
\beq \label{ssi}
\usigma = \min(a/{\xi}, \, m)
\eeq
with respect to the norm $\|\cdot\|_\cU$.

In view of \eqnok{lipconst}, \eqnok{DU1}, \eqnok{ssi} and
Corollary \ref{mtm-sm1}, it follows that
the iteration-complexity of the variant of Nesterov's smooth method for 
finding an $\epsilon$-optimal solution of \eqnok{constr-sadpt2d-1} and
its dual is bounded by
\[
\Gamma (\xi) = \left \lceil \frac{2\gamma m \|\Lambda ^{-1}\|}
{\sqrt{\epsilon}}\sqrt{\frac{2[1+ \log(n/m)]}
{\min(a/{\xi}, \, m)}} \, \right \rceil.
\] 

As a consequence of the above discussion and Corollary \ref{mtm-sm1},
we obtain the following result.

\begin{theorem} \label{compl-1-1}
For a given $\epsilon>0$, the variant of Nesterov's smooth method,
with prox-function defined by \eqnok{du-1}-\eqnok{abc}, $L$ given by 
\eqnok{lipconst} and $\usigma$ given by \eqnok{ssi} with $\xi=a/m$ , 
applied to \eqnok{constr-sadpt2d-1}, finds an $\epsilon$-optimal solution 
of problem \eqnok{constr-sadpt2d-1} and its dual
in a number of iterations which does not exceed
\beq \label{sma-compl1}
\left\lceil 
\frac{2\sqrt{2}\gamma\|\Lambda ^{-1}\|
\sqrt{m}}{\sqrt{\epsilon}}\sqrt{1+\log(n/m)}
\right\rceil.
\eeq  
\end{theorem}

\begin{proof}
We have seen in the discussion preceding this theorem that
the iteration-complexity of the variant of Nesterov's smooth method for 
finding an $\epsilon$-optimal solution of \eqnok{constr-sadpt2d-1} and
its dual is bounded by $\Gamma(\xi)$ for any $\xi>0$. Taking $\xi=a/m$, 
we obtain the iteration-complexity bound \eqnok{sma-compl1}.
\end{proof}

\vgap

We observe that the iteration-complexity given in \eqnok{sma-compl1} is in 
terms of the transformed data of problem \eqnok{multi-reg4}. We next relate 
it to the original data of problem \eqnok{multi-reg4}. The proof of the 
following corollary is similar to that of Corollary \ref{compl-3-2}.

\begin{corollary} \label{compl-1-2}
For a given $\epsilon>0$, the variant of Nesterov's smooth method,
with prox-function defined by \eqnok{du-1}-\eqnok{abc}, $L$ given by 
\eqnok{lipconst} and $\usigma$ given by \eqnok{ssi} with $\xi=a/m$ , applied to 
applied to \eqnok{constr-sadpt2d} finds an $\epsilon$-optimal solution
of problem \eqnok{constr-sadpt2d} and its dual
in a number of iterations which does not exceed
\beq \label{compl-constr}
\left\lceil
\frac{2\sqrt{2}\gamma\|(A^TA)^{-1/2}\| \sqrt{m}}{\sqrt{\epsilon}}
\sqrt{\log (n/m)+1}
\right\rceil.
\eeq
\end{corollary}

Observe that, in view of Lemma \ref{penalty-lem} with $\bar X=0$
and Theorem \ref{exactpenal-thm2},
\eqnok{compl-constr} is also an iteration-complexity bound
for finding an $\epsilon$-optimal solution of problem
\eqnok{constr} whenever
\[
\gamma = \frac{\|H\|^2_F}{M} = \frac{\| (A^TA)^{-1/2}A^TB\|^2_F}{M},
\]
where the later equality is due to \eqnok{transform}.

Based on the discussion below and in Subsection \ref{sma-implement-2},
the arithmetic operation cost per iteration of the variant of
Nesterov's smooth method when applied to problem
\eqnok{constr-sadpt2d} is bounded by ${\cal O}(mpq)$ where
$m=\min(p,q)$, due to the fact that its most expensive operation consists of
finding a partial singular value decomposition of a $p \times q$ matrix $h$
as in \eqnok{singdec}.
Thus, the overall arithmetic-complexity of the variant of
Nesterov's smooth method when applied to \eqnok{constr-sadpt2d} is 
\[
{\cal O}\left(\frac{\gamma\|(A^TA)^{-1/2}\|}{\sqrt{\epsilon}}m^{3/2}
pq\sqrt{\log (n/m)}\right).
\] 

After having completely specified all the ingredients required by the variant 
of Nesterov's smooth method for solving \eqnok{constr-sadpt2d-1},
we now discuss some of the computational technicalities involved in the
actual implementation of the method.

First, for a given $u\in U$, the optimal solution
for the maximization subproblem \eqnok{fu} needs to be found
in order to compute the gradient of $\nabla f(u)$.
The details here are similar to the corresponding ones described
in Subsection \ref{sma-implement-2} (see the paragraph containing
relation \eqnok{qp-ball}).

In addition, each iteration of the variant of Nesterov's smooth method
requires solving subproblem \eqnok{proxsub1}.
In view of \eqnok{gu} and \eqnok{phiuv2}, it 
is easy to observe that for every $u=(t,W)\in U$,
we have $\nabla f(u) = (\eta, \cG(v))$ for some $\eta \in \Re$ and
$v\in\Re^{p \times q}$.
Also, by \eqnok{du-1}, \eqnok{abc} and \eqnok{u0-2}, we easily see that
$\nabla \dU(u_0) = (\log n -1) (1, -I)$. 
Using these results along with \eqnok{Bregdist} and \eqnok{dv3},
we easily see that subproblem \eqnok{proxsub1} is equivalent to
one of the form
\beq \label{prox-subprob1}
\min\limits_{(t,W)\in\tOmega}
\left\{\left(\varsigma I + \cG(h)\right) \bullet W + \alpha t +
\tr(W\log W) + a t \log t \right\} 
\eeq
for some $\alpha, \varsigma \in \Re$ and $h\in\Re^{p \times q}$, where 
$a$ and $\tOmega$ are given by \eqnok{abc} and \eqnok{tOmega},
respectively.

We now discuss how the above problem can be efficiently solved.
First, note that by \eqnok{tOmega}, we have
$(t,W) \in \tilde \Omega$ if, and only if,
$W=tW'$ for some $W' \in \Omega_1$. This observation together with
the fact that $\tr W' =1$ for every $W' \in \Omega_1$ allows us to
conclude that problem \eqnok{prox-subprob1} is equivalent to
\begin{align}
\min\limits_{W'\in\Omega_1, \, t \in [0,1]} &
\left\{ t \left(\varsigma I + \cG(h)\right) \bullet W' + \alpha t +
t \left[ \tr(W' \log W')  + (\log t) \tr(W') \right] + a t \log t \right\}
\label{mintW} \\
= & \min_{t \in [0,1]} \alpha t + (a+1) t \log t + t d, \label{mint}
\end{align}
where
\beq \label{minW'}
d:= \min\limits_{W'\in\Omega_1}
\left(\varsigma I + \cG(h)\right) \bullet W' + \tr(W' \log W').
\eeq
Moreover, if $W'$ is the optimal solution of \eqnok{minW'} and
$t$ is the optimal solution of \eqnok{mint}, then $W=tW'$ is the
optimal solution of \eqnok{mintW}. Problem \eqnok{minW'}
is of the form \eqnok{fg-subprob} where an efficient scheme
for solving it is described in Subsection \ref{sma-implement-2}.
It is easy to see that the optimal solution of \eqnok{mint} is
given by
\[
t = \min \left[ 1 \, , \, \exp \left(-1- \frac{\alpha+d}{a+1} \right) \right].
\]

Finally, to terminate the variant of Nesterov's smooth approximation scheme, 
we need to properly evaluate the primal and dual objective functions of problem 
\eqnok{constr-sadpt2d-1} at any given point. As seen from \eqnok{fu}
and \eqnok{phiuv2}, the primal objective function $f(u)$ of
\eqnok{constr-sadpt2d-1} can be computed by solving a problem in the form
of \eqnok{qp-ball}. Additionally, 
in view of \eqnok{char4} and \eqnok{tOmega}, the dual objective function $g(v)$ 
of \eqnok{constr-sadpt2d-1} can be computed as
\[
g(v) = -\frac12\|\trx \Lambda v - \K\|^2_F 
- \gamma \left[\trx\sum\limits_{i=1}^m \sigma_i(v)-M\right]^+,
\ \ \forall v \in V.
\] 

\section{Computational results}
\label{comp} 


In this section, we report the results of our computational experiment which 
compares the performance of the variant of of Nesterov's smooth method discussed 
in Subsection \ref{sma-implement-2} for solving problem \eqnok{unconstr} with the 
interior point method implemented in SDPT3 version 4.0 (beta) \cite{ToTuTo06-1} 
on a set of randomly generated instances.

The random instances of \eqnok{unconstr} used in our experiments were generated as follows. 
We first randomly generated matrices $A\in\Re^{l \times p}$ and $B\in
\Re^{l \times q}$, where $p=2q$ and $l=10q$, with entries uniformly distributed
in $[0,1]$ for different values of $q$.
We then computed $H$ and $\Lambda$ for \eqnok{unconstr} according to the
procedures described in Subsection 
\ref{reduction} and set the parameter $\lambda$ in \eqnok{unconstr} to one.
In addition, all computations were performed on an Intel 
Xeon 5320 CPU (1.86GHz) and 12GB RAM running Red Hat Enterprise Linux 4 
(kernel 2.6.9).
 

In this experiment,  we compared the performance of the variant of Nesterov's smooth 
method (labeled as VNS) discussed in Subsection \ref{sma-implement-2} for solving 
problem \eqnok{unconstr} with the interior point method implemented in SDPT3 
version 4.0 (beta) \cite{ToTuTo06-1} for solving the cone programming reformulation 
\eqnok{cone-repr1}. The code for VNS is written in C, and the initial point for this 
method is set to be $u_0=I/(p+q)$. It is worth mentioning that the code SDPT3 uses MATLAB 
as interface to call several C subroutines to handle all its heavy computational tasks.
SDPT3 can be suitably applied to solve a standard cone programming with the underlying 
cone represented as a Cartesian product of nonnegative orthant, second-order cones, and 
positive semidefinite cones. The method VNS terminates once the duality gap is less than 
$\epsilon=10^{-8}$, and SDPT3 terminates once the relative accuracy is less than $10^{-8}$. 

The performance of VNS and SDPT3 for our randomly generated instances 
are presented in Table \ref{result-3}. The problem size $(p,q)$ is given in 
column one. The numbers of iterations of VNS and SDPT3 are given in columns 
two and three, and the objective function values are given in columns four and 
five, CPU times (in seconds) are given in columns six to seven, and the 
amount of memory (in mega bytes) used by VNS and SDPT3 are given in the last two 
columns, respectively. The symbol ``N/A'' means ``not available''. The computational 
result of SDPT3 for the instance with $(p,q)=(120,60)$ is not available since it 
ran out of the memory in our machine (about 15.73 giga bytes). We conclude from 
this experiment that the method VNS, namely, the variant of Nesterov's smooth method, 
generally outperforms SDPT3 substantially even for relatively small-scale problems. 
Moreover, VNS requires much less memory than SDPT3. For example, for the instance with 
$(p,q)=(100,50)$, SDPT3 needs $10445$ mega ($\approx 10.2$ giga) bytes of memory, 
but VNS only requires about $4.23$ mega bytes of memory; for the instance with 
$(p,q)=(120,60)$, SDPT3 needs at least $16109$ mega ($\approx 15.73$ giga) bytes 
of memory, but VNS only requires about $4.98$ mega bytes of memory.
    
\begin{table}[t]
\caption{Comparison of VNS and SDPT3}
\centering
\label{result-3}
\begin{small}
\begin{tabular}{|c||rr||rr||rr|rr|}
\hline 
\multicolumn{1}{|c||}{Problem} & \multicolumn{2}{c||}{Iter} &  
\multicolumn{2}{c||}{Obj} & \multicolumn{2}{c||}{Time} &  
\multicolumn{2}{c|}{Memory} \\
\multicolumn{1}{|c||}{(p, q)} & \multicolumn{1}{c}{\sc VNS} 
& \multicolumn{1}{c||}{\sc SDPT3} & \multicolumn{1}{c}{\sc VNS} 
& \multicolumn{1}{c||}{\sc SDPT3} & \multicolumn{1}{c}{\sc VNS} 
& \multicolumn{1}{c||}{\sc SDPT3} & \multicolumn{1}{c}{\sc VNS} 
& \multicolumn{1}{c|}{\sc SDPT3} \\
\hline
(20, 10) & 36145 & 17 & 4.066570508 & 4.066570512  & 16.6 & 5.9 & 2.67 & 279 \\
(40, 20) & 41786 & 15 & 8.359912031 & 8.359912046 & 55.7 & 77.9  & 2.93 & 483 \\
(60, 30) & 35368 & 15 & 13.412029944 & 13.412029989 & 96.7 & 507.7 & 3.23 & 1338 \\
(80, 40) & 36211 & 15 & 17.596671337 & 17.596671829 &  182.9 & 2209.8 & 3.63 & 4456 \\
(100, 50) & 33602 & 19 & 22.368563640 & 22.368563657 & 272.6 & 8916.1 & 4.23 & 10445 \\
(120, 60) & 33114 & {N/A} & 26.823206950 & {N/A} & 406.6 & {N/A} & 4.98 & $>16109$ \\
\hline
\end{tabular}
\end{small}
\end{table}

\section{Concluding remarks}
\label{concl-remark}    

In this paper, we studied convex optimization methods for computing the 
trace norm regularized least squares estimate in multivariate linear 
regression. In particular, we explore a variant of Nesterov's smooth method 
proposed by Tseng \cite{tseng08} and interior point methods for computing the 
penalized least squares estimate. The performance of these methods is then 
compared using a set of randomly generated instances. We showed that
the variant of Nesterov's smooth method generally substantially outperforms 
the interior point method implemented in SDPT3 version 4.0 (beta) \cite{ToTuTo06-1}. 
Moreover, the former method is much more memory efficient. 


In Subsection \ref{reduction} we provided an approach for simplifying problem 
\eqnok{multi-reg4} which changes the variable $U$, in addition to
the data $A$ and $B$. A drawback of this approach is that it can not
handle extra constraints (not considered in this paper) on $U$.
It turns out that there exists an alternative scheme for simplifying 
problem \eqnok{multi-reg4}, i.e. one that eliminates the dependence of the data
on the (generally, large) dimension $l$, which does not change $U$.
Indeed, by performing either a QR factorization
of $A$ or a Cholesky factorization of $A^TA$, compute
an upper triangular matrix $R$ such that $R^TR=A^TA$.
Letting $G:=R^{-T}A^TB$, it is straightforward to show that problem
\eqnok{multi-reg4} can be reduced to 
\beq \label{new-form}
\min\limits_U \left\{ \|G-RU\|^2_F: \ \sum\limits_{i=1}^m \sigma_i(U) \le M\right\}.
\eeq
Clearly, in contrast to reformulation \eqnok{constr}, the above one
does not change the variable $U$ and hence extra constraints on $U$ can be
easily handled.
On the other hand, a discussion similar to that in
Subsection \ref{sma-implement-2} shows that each iteration of the variant of Nesterov's
smooth method applied to \eqnok{new-form}, or its Lagrangian
relaxation version, needs to solve
subproblem \eqnok{qp-ball} with $\Lambda$ replaced by $R$.
Since $R$ is an upper triangular matrix and 
$\Lambda$ is a diagonal matrix, the later subproblems 
are much harder to solve than subproblems of the form \eqnok{qp-ball}.
For this reason, we have opted to use
reformulation \eqnok{constr} rather than \eqnok{new-form} in this paper.    

\section*{Appendix}

In this section, we discuss some technical results that are used
in our presentation. More specifically, we discuss two ways of solving
a constrained nonlinear programming problem based on some
unconstrained nonlinear programming reformulations.

Given a set $\emptyset \neq \setX \subseteq \Re^n$ and
functions $f: \setX \to \Re$ and $h: \setX \to \Re^k$,
consider the nonlinear programming problem:
\beq \label{convex-opt1} 
f^* = \inf \, \{ f(x): x \in \setX, \, h_i(x) \le 0, \, i=1,\ldots,k\}.
\eeq
The first reformulation of \eqnok{convex-opt1} is based on
the exact penalty approach, which consists of solving
the exact penalization problem
\beq \label{penalty} 
\fg^* = \inf \, \{ \fg (x) := f(x) + \gamma [g(x)]^+ : x \in \setX \},
\eeq
for some large penalty parameter $\gamma >0$,
where $g(x) = \max \{h_i(x) : i=1,\ldots,k\}$.
To obtain stronger consequences, we make the following
assumptions about problem \eqnok{convex-opt1}:
\bi
\item[\bf A.1)] The set $\setX$ is convex and functions $f$ and $h_i$ are
convex for each $i=1,\ldots,k$;
\item[\bf A.2)] $f^* \in \Re$ and
there exists a point $\bx\in \setX$ such that $g(\bx) < 0$.
\ei

We will use the following notion throughout the paper.

\begin{definition} \rm
Consider the problem of minimizing a real-valued function $f(x)$
over a certain nonempty feasible region ${\cal F}$ contained in the
domain of $f$ and let $\bar f := \inf \{ f(x) : x \in {\cal F} \}$.
For $\epsilon \ge 0$, we say that $x_{\epsilon}$ is an
{\bf $\epsilon$-optimal solution} of
this problem if $x_{\epsilon} \in {\cal F}$ and
$f(x_{\epsilon}) \le \epsilon + \bar f$.
\end{definition}

We note that the existence of an $\epsilon$-optimal solution for
some $\epsilon>0$ implies that $\bar f$ is finite.


\begin{theorem} \label{penalty-thm}
Suppose Assumptions A.1 and A.2 hold and define
\[
\bar \gamma :=\frac{f(\bx)-f^*}{|g(\bx)|} \ge 0.
\]
For $x \in \setX$, define
\beq \label{ztheta}
z(x):=\frac{x+\theta(x)\bx}{1+\theta(x)}, \ \ \ \
\mbox{where} \ \ \theta(x):=\frac{[g(x)]^+}{|g(\bx)|}.
\eeq
Then, the following statements hold:
\begin{itemize}
\item[a)]
for every $x \in \setX$, the point $z(x)$ is a feasible solution of
\eqnok{convex-opt1};
\item[b)]
$\fg^* = f^*$ for every $\gamma \ge \bar \gamma$;
\item[c)]
for every $\gamma \ge \bar \gamma$ and $\epsilon \ge 0$,
any $\epsilon$-optimal solution of \eqnok{convex-opt1} is
also an $\epsilon$-optimal solution of \eqnok{penalty};
\item[d)]
if $\gamma \ge \bar \gamma$, $\epsilon \ge 0$
and $\xeps$ is an $\epsilon$-optimal solution of
\eqnok{penalty},
then the point $z(\xeps)$
is an $\epsilon$-optimal solution of \eqnok{convex-opt1}.
\item[e)]
if $\gamma > \bar \gamma$, $\epsilon \ge 0$
and $\xeps$ is an $\epsilon$-optimal solution of
\eqnok{penalty}, then
$f(\xeps) - f^* \le \epsilon$ and
$[g(\xeps)]^+ \le \epsilon/(\gamma-\bar \gamma)$.
\end{itemize}
\end{theorem}

\begin{proof}
Let $x \in \setX$ be arbitrarily given.
Clearly, convexity of $\setX$, the assumption that $x^0 \in \setX$
and the definition of $z(x)$ imply that
$z(x) \in \setX$. Moreover, Assumption A.1 implies that $g:X \to \Re$ is
convex. This fact, the assumption that $g(x^0)<0$, and the
definitions of $z(x)$ and $\theta(x)$ then imply that
\[
g(z(x)) \le \frac{g(x)+\theta(x) g(\bx)}{1+\theta(x)} \le 
\frac{[g(x)]^+ -\theta(x) |g(\bx)|}{1+\theta(x)} = 0.
\]
Hence, statement (a) follows.

To prove statement (b),
assume that $\gamma \ge \bar \gamma$ and let $x \in \setX$ be given.
Convexity of $f$
yields $(1+\theta(x))f(z(x)) \le f(x)+\theta(x)f(\bx)$,
which, together with the definitions of $\bar \gamma$ and
$\theta(x)$, imply that
\beqa
\fg(x) - f^* & = &
f(x) + \gamma [g(x)]^+ - f^* \nn \\ [3pt]
&\ge& (1+\theta(x))f(z(x)) - \theta(x) f(\bx) + \gamma [g(x)]^+ - f^*
\nn \\ [3pt]
&=& (1+\theta(x))(f(z(x))-f^*) - \theta(x) (f(\bx)-f^*)
+ \gamma [g(x)]^+ \nn \\ [3pt]
&=& (1+\theta(x))(f(z(x))-f^*) + (\gamma-\bar \gamma) [g(x)]^+.
\label{easyap1}
\eeqa
In view of the assumption that $\gamma \ge \bar \gamma$ and
statement (a), the above inequality implies that
$\fg(x) - f^* \ge 0$ for every $x \in \setX$, and hence
that $\fg^* \ge f^*$. Since the inequality $\fg^* \le f^*$ obviously
holds for any $\gamma \ge 0$, we then conclude that $\fg^* = f^*$
for any $\gamma \ge \bar \gamma$.
Statement (c) follows as an immediate consequence of (b).

For some $\gamma \ge \bar \gamma$ and $\epsilon\ge 0$, assume now
that $\xeps$ is an $\epsilon$-optimal solution of \eqnok{penalty}.
Then, statement (b) and inequality \eqnok{easyap1} imply that
\beq \label{easyap2}
\epsilon \ge \fg(\xeps) - \fg^*
\ge (1+\theta(\xeps))(f(z(\xeps))-f^*) + (\gamma-\bar \gamma) [g(\xeps)]^+.
\eeq
Using the assumption that $\gamma \ge \bar \gamma$, the above
inequality clearly implies that
$f(z(\xeps)) - f^* \le \epsilon/(1+\theta(\xeps)) \le \epsilon$, and hence that
$z(\xeps)$ is an $\epsilon$-optimal solution of \eqnok{convex-opt1}
in view of statement (a). Hence, statement (d) follows.
Moreover, if $\gamma > \bar \gamma$,
we also conclude from \eqnok{easyap2} that $[g(\xeps)]^+ \le \epsilon / (\gamma-\bar \gamma)$.
Also, the first inequality of \eqnok{easyap2} implies that
$f(\xeps) - f^* \le f(\xeps) + \gamma[g(\xeps)]^+ - f^* =
\fg(\xeps) - \fg^* \le \epsilon$, showing that statement (e) holds.
\end{proof}

\vgap

We observe that the threshold value $\bar \gamma$ depends on the
optimal value $f^*$, and hence can be computed only for those
problems in which $f^*$ is known. If instead a lower bound $f_l \le f^*$
is known, then choosing the penalty parameter $\gamma$ in
problem \eqnok{penalty} as $\gamma:=(f(x^0)-f_l)/|g(x^0)|$ guarantees
that an $\epsilon$-optimal solution $\xeps$ of \eqnok{penalty} yields the
$\epsilon$-optimal solution $z(\xeps)$ of \eqnok{convex-opt1},
in view of Theorem \ref{penalty-thm}(c).

The following result, which is a slight variation of a result due to
H.\ Everett (see for example pages 147 and 163 of \cite{HiLe93-1}),
shows that approximate optimal solutions of Lagrangian
subproblems associated with \eqnok{convex-opt1} yield
approximate optimal solutions of a perturbed version of \eqnok{convex-opt1}.

\begin{theorem} (Approximate Everett's theorem) \label{lagr-eps}
Suppose that for some $\lambda \in \Re^k_{+}$ and $\epsilon \ge 0$,
$\xleps$ is an $\epsilon$-optimal solution of the problem
\beq \label{lagr-prob}
f^*_{\lambda} = \inf \, \left\{ f(x) + \sum_{i=1}^k \lambda_i h_i(x)
: x \in \setX \right\}.
\eeq
Then, $\xleps$ is an $\epsilon$-optimal solution of the problem 
\beq \label{constr-prob}
f^*_{\epsilon\lambda} = \inf \, \left\{ f(x) : x \in \setX, \,
h_i(x)\le h_i(\xleps), \, i=1,\ldots,k \right\}.
\eeq
\end{theorem}

\begin{proof}
Let $\tilde x$ be a feasible solution of \eqnok{constr-prob}.
Since $\xleps$ is an $\epsilon$-optimal solution of \eqnok{lagr-prob}, 
we have $f(\xleps) + \sum_{i=1}^k \lambda_i h_i(\xleps) \le f^*_{\lambda} + \epsilon$.
This inequality together with the definition of $f^*_{\lambda}$ in
\eqnok{lagr-prob} implies that
\beqas
f(\xleps) &\le& f^*_{\lambda} - \sum_{i=1}^k \lambda_i h_i(\xleps) + \epsilon \ \le \
f(\tilde x) + \sum_{i=1}^k \lambda_i [h_i(\tilde x)-h_i(\xleps)]+ \epsilon
\ \le \ f(\tilde x) + \epsilon,
\eeqas
where the last inequality is due to the fact that $\lambda_i \ge 0$ for
all $i = 1,\ldots,k$ and $\tilde x$ is feasible solution of \eqnok{constr-prob}.
Since the latter inequality holds for every feasible solution $\tilde x$ of
\eqnok{constr-prob}, we conclude that
$f(\xleps) \le f^*_{\epsilon\lambda} + \epsilon$, and hence that
$\xleps$ is an $\epsilon$-optimal solution of \eqnok{constr-prob}.
\end{proof}

\vgap

If our goal is to solve problem
$\inf \{ f(x) : x \in \setX, \, h_i(x) \le b_i,\, i=1,\ldots,k \}$
for many different
right hand sides $b \in \Re^k$, then, in view of the above result,
this goal can be accomplished by minimizing the Lagrangian subproblem
\eqnok{lagr-prob} for many different Lagrange multipliers
$\lambda \in \Re^k_+$. We note that this idea is specially popular
in statistics for the case when $k=1$.

\section*{Acknowledgements}
The authors would like to thank two anonymous referees and the associate editor 
for numerous insightful comments and suggestions, which have greatly improved 
the paper.

\end{document}